\documentclass[11pt,reqno]{amsart}

\pdfoutput = 1

\usepackage[english]{babel}
\usepackage{amsmath,amsfonts,amssymb,amsthm}
\usepackage{amsfonts}
\usepackage{amsxtra}

\usepackage{array,dcolumn}
\usepackage{graphicx}
\usepackage[hyperfootnotes=true,
        pdffitwindow=true,
        plainpages=false,
        pdfpagelabels=true,
        pdfpagemode=UseOutlines,
        pdfpagelayout=SinglePage,
        hyperindex,]{hyperref}

\usepackage[T1]{fontenc}
\usepackage[utf8]{inputenc}
\usepackage[activate={true,nocompatibility},spacing,kerning]{microtype}
\usepackage[sc,osf]{mathpazo}
\microtypecontext{spacing=nonfrench}

 \calclayout

\newcommand{\mathsym}[1]{{}}
\newcommand{\unicode}[1]{{}}

\numberwithin{equation}{section}

\theoremstyle{plain}
\newtheorem{theorem}{Theorem}

\newtheorem{corollary}[theorem]{Corollary}
\newtheorem{proposition}[theorem]{Proposition}
\numberwithin{theorem}{section}

\theoremstyle{definition}

\theoremstyle{remark}
\newtheorem{remark}[theorem]{Remark}

\begin{document}
\title[Circular $\beta$-ensemble distributions and joint moments]{Joint moments of a characteristic polynomial and its derivative for
the circular $\beta$-ensemble}
\author{Peter J. Forrester}
\address{School of Mathematics and Statistics, 
ARC Centre of Excellence for Mathematical
 and Statistical Frontiers,
University of Melbourne, Victoria 3010, Australia}
\email{pjforr@unimelb.edu.au}

\begin{abstract}
The problem of calculating the scaled limit of the joint moments of the characteristic polynomial,
and the derivative of the characteristic polynomial, for matrices from the unitary group with Haar
measure first arose in studies relating to the Riemann zeta function in the thesis of Hughes.
Subsequently, Winn showed that these joint moments can equivalently be written as the moments
for the distribution of the trace in the Cauchy unitary ensemble, and furthermore relate to certain hypergeometric functions
based on Schur polynomials, which enabled explicit computations.
 We give a $\beta$-generalisation of these results, where now the role of the Schur
polynomials is played by the Jack polynomials. This leads to an explicit evaluation of the scaled moments and
the trace distribution for all $\beta > 0$,
subject to the constraint that a particular parameter therein is equal to a non-negative integer. Consideration
is also given to the calculation of the moments of the singular statistic $\sum_{j=1}^N 1/x_j$ for the Jacobi
$\beta$-ensemble.
\end{abstract}
\maketitle

\section{Introduction}\label{s1}
\subsection{The particular circular $\beta$-ensemble joint moments and statement of their scaled large $N$ form}
With motivations in the application of random matrix theory to the study of statistical properties
of the Riemann zeta function in accordance with the
hypothesis of Keating and Snaith \cite{KS00a}, Hughes \cite{Hu01} considered the quantity
\begin{equation}\label{1.1}
\mathcal F_N(s,h) = \int_{U} \Big |V_U(\theta)  |_{\theta = 0} \Big |^{2(s-h)}  
\Big | {d \over d\theta} V_U(\theta) \Big |_{\theta = 0}  \Big |^{2h} \, d \mu_N(U), \qquad - {1 \over 2} < h < s + {1 \over 2}.
\end{equation}
In (\ref{1.1}), $U \in U(N)$, the group of $N \times N$ unitary matrices, which are to be chosen with
Haar measure. The latter is indicated in (\ref{1.1}) by the symbol $d \mu_N(U)$, which furthermore is
to be normalised so that its integration over $U$ gives unity. With $\{ e^{i \theta_j} \}_{j=1}^N$ denoting the
eigenvalues of $U$, the function $V_U(\theta)$ appearing in (\ref{1.1}) is specified by
\begin{equation}\label{1.2}
V_U(\theta) = e^{i N (\theta + \pi)/2 - i \sum_{j=1}^N \theta_j/2 } \prod_{j=1}^N (1 - e^{- i (\theta - \theta_j)});
\end{equation}
note in particular that (\ref{1.2}) is real for $\theta$ real, and its absolute value is equal to the
absolute value of the characteristic polynomial of $U$.

The integrand of (\ref{1.1}) is a function of the eigenvalues of $U$ only. Since the work of
Weyl \cite{We39} --- see the review \cite{DF17} for historical context --- it has been known that the
Haar measure $d \mu_N(U)$, upon changing variables to the eigenvalues and eigenvectors,
decomposes to an eigenvalue dependent factor, and an eigenvector dependent factor.
This allows the latter to be integrated over, reducing (\ref{1.1}) to read
\begin{multline}\label{1.3}
\mathcal F_N(s,h) =  {1 \over N! (2 \pi)^N} \int_0^{2 \pi} d \theta_1 \cdots  \int_0^{2 \pi} d \theta_N \,  \Big |V_U(\theta)  |_{\theta = 0}  \Big |^{2(s-h)} 
\Big | {d \over d\theta} V_U(\theta) \Big |_{\theta = 0}  \Big |^{2h} \\
\times 
 \prod_{1 \le j < k \le N} | e^{i \theta_k}  - e^{i \theta_j} |^2.
\end{multline}

There is a natural generalisation of (\ref{1.3}) obtained by replacing the power of $2$ in the product
over $j < k $ by the power of $\beta$,
\begin{multline}\label{1.3a}
\mathcal F_{N,\beta}(s,h) =  {1 \over C_{N,\beta} } \int_0^{2 \pi} d \theta_1 \cdots  \int_0^{2 \pi} d \theta_N \,  \Big |V_U(\theta)  |_{\theta = 0} \Big |^{2(s-h)} 
\Big | {d \over d\theta} V_U(\theta) \Big |_{\theta = 0}  \Big |^{2h} \\
\times 
  \prod_{1 \le j < k \le N} | e^{i \theta_k}  - e^{i \theta_j} |^\beta,
\end{multline}
where
\begin{equation}\label{1.3b}
C_{N,\beta} = (2 \pi)^N {\Gamma(\beta N/2 + 1) \over (\Gamma(1+\beta/2))^N}.
\end{equation}
When multiplied by $1/ C_{N,\beta}$, the product involving $\beta$ in (\ref{1.3a}) for $\beta = 1$ and 4  is the eigenvalue probability density
function (PDF) for 
particular random matrix ensembles constructed out of unitary matrices $U$ with Haar measure. For $\beta = 1$, the construction
consists of forming $U^T U$; for $\beta = 4$, $U$ of size $2N \times 2N$, and $Z_{2N}$ the $2N \times 2N$ anti-symmetric
tridiagonal with entries $-1$ ($1$) on the leading upper (lower) sub-diagonal, the construction is to form $U^D U$ where
$U^D = Z_{2N}^{-1} U^T Z_{2N}$ and to consider only distinct eigenvalues (in this case the spectrum is doubly degenerate).
These two ensembles, together with the Haar distributed unitary matrices themselves, are the three circular ensembles of Dyson \cite{Dy62}.
Many years later it shown in \cite{KN04} how to construct certain random unitary Hessenberg matrices which give rise to the
eigenvalue PDF implied by (\ref{1.3}) for general $\beta > 0$. We remark that the parameter $\beta$ is often referred to as
the Dyson index.

Our primary interest in the present paper is to extend a number of identities associated with (\ref{1.3}),
which culminate in an explicit determination of the scaled large $N$ limit for $s$ a positive integer and continuous $h$,
 to the general $\beta > 0$ case (\ref{1.3a}). 
The known identities  can be found in the work of Winn \cite{Wi12}, as can the determination of the large $N$ limit,
provided $h$ is restricted to non-negative integer or half integer values. The very recent works \cite{AKW20,ABGS20a} give a
viewpoint which allows the working in \cite{Wi12} to be adapted to continuous values of $h$.
Applying this to the present setting,
our results in the general $\beta > 0$ case similarly allow for the large $N$ form of (\ref{1.3a}) for continuous $h$ to be determined.

\begin{theorem}\label{P1A}
With $s$ a non-negative integer, let $\kappa= (\kappa_1,\dots,\kappa_s)$ denote a partition of $s$ parts. Define the generalised Pochhammer symbol by
(\ref{C0}) below. Let $C_\kappa^{(\alpha)}(x_1,\dots.x_s)$ denote the Jack polynomial indexed by the partition $\kappa $ and
depending on the parameter $\alpha$, normalised so that $C_\kappa^{(\alpha)}((1)^s) :=
C_\kappa^{(\alpha)}(x_1,\dots.x_s) |_{x_1 = \cdots = x_s = 1}$ is given by (\ref{2.6+}) below. For $\beta > 0$ and
with  
\begin{equation}\label{R}
-1/2 < {\rm Re} \, h  < s + 1/2
 \end{equation}
 we have
 \begin{multline}\label{1.91}
 \lim_{N \to \infty} {1 \over N^{2 s^2/\beta + 2h} }\mathcal{F}_{N,\beta}(s,h) =
 {1 \over 2^{2h} } 
    \prod_{j=1}^s  {\Gamma(2j/\beta) \over \Gamma(2(s+j)/\beta)} \\
  \times  {1 \over \cos \pi h} \sum_{\kappa} {(-2 h)_{|\kappa|} \over | \kappa |! [4 s/ \beta]_\kappa^{(\beta/2)} }
 \Big ( {4 \over \beta} \Big )^{|\kappa|} C_\kappa ^{(\beta/2)}  ( (1)^s ).
 \end{multline}
\end{theorem}

\begin{remark}
1.~In the case $s = 0$ the sum over $\kappa$ is empty, and (\ref{1.91}) reduces to
\begin{equation}\label{1.92}
 \lim_{N \to \infty} {1 \over N^{ 2h} }\mathcal{F}_{N,\beta}(s,h)  \Big |_{s=0} =   {1 \over 2^{2h} }  {1 \over \cos \pi h} , \quad -1/2 < h < 1/2,
 \end{equation}
independent of $\beta$. \\
2.~In the case $s = 1$, $ C_\kappa ^{(\beta/2)}  ( (1)^s ) = 1$ and the  generalised Pochhammer symbols reduce to their
classical counterparts. Then (\ref{1.91}) simplifies to read
\begin{equation}\label{1.92a}
 \lim_{N \to \infty} {1 \over N^{(2/\beta) +  2h} }\mathcal{F}_{N,\beta}(s,h)  \Big |_{s=1} =   {1 \over 2^{2h} } { \Gamma(2/\beta) \over \Gamma(4/\beta) }  {1 \over \cos \pi h} 
 \, {}_1 F_1 (-2h; 4/\beta;4/\beta), \quad -1/2 < h < 3/2.
 \end{equation}
 At $h=1/2$ this is to be computed as a limit to this value, which gives
 \begin{align}\label{1.92b}
 \lim_{N \to \infty} {1 \over N^{(2/\beta) +  1} } & \mathcal{F}_{N,\beta}(s,h)  \Big |_{s=1, h= 1/2}  \nonumber \\ & =    { \Gamma(2/\beta) \over \Gamma(4/\beta) } \bigg ( {1 \over \pi} - {1 \over 2 \pi} {1 \over 1 + \beta/4} \,
 {}_2 F_2(1,1;3,4/\beta+2;4/\beta)  \bigg )  \nonumber \\
 & = \begin{cases}
 (3 e^4 - 103)/(768 \pi) ,& \beta = 1, \\
 (e^2 - 5)/(4 \pi), & \beta = 2, \\
 (e - {2})/\sqrt{\pi}, & \beta = 4.
 \end{cases}
 \end{align} 
 Here the explicit values have been determined from closed forms of
 ${}_2 F_2$ found using computer algebra. The value for $\beta = 2$ was first derived in this context by
 Winn \cite{Wi12}, and most  recently by Assiotis et al.~\cite{ABGS20a} using ideas extending those in  \cite{Wi12}. In the 
 setting of joint moments of the Hardy function in the theory of the Riemann zeta function motivating the
 study of (\ref{1.1}), the value $ (e^2 - 5)/(4 \pi)$ was earlier
 known from the work of Conrey and Ghosh \cite{CG84}. Combining this, and the Keating--Snaith hypothesis relating averages of the
 Riemann zeta function to averages over $U(N)$, the statement of (\ref{1.92b}) with $\beta = 2$ first appeared as a conjecture
 in the thesis of Hughes \cite{Hu01}. \\
 3.~For $h$ a non-negative integer in (\ref{1.91}) the sum over partitions can
be restricted to partitions of size $|\kappa| \le 2h$, and so terminates. In the case $h=0$ the sum is empty and
is to be interpreted as equal to unity. The resulting formula then corresponds to the leading asymptotics of
the moments $\langle |V_U(\theta) |_{\theta = 0}|^{2s} \rangle$,
where the average is with respect to the circular $\beta$-ensemble, and is consistent with results known from the earlier works  \cite{BF97f,KS00a}.
The second line of  (\ref{1.91}) is in fact a rational function of $s$ for $h$ a non-negative integer. We show in
Remark \ref{R3.6} below that the validity of this rational function formula extends to all $s > h - 1/2$.
 \end{remark}

 As a secondary point, discussed in \S \ref{S4},
we identify a structurally very similar class of moments to $\mathcal F_{N,\beta}(s,h)$ for the
Jacobi $\beta$-ensemble, although now parametrised by a single exponent $h$, rather than
two exponents $(s,h)$  as in   (\ref{1.3a}). These are specified by
\begin{multline}\label{1.7+}
\mathcal{H}_N^{(J)}(h) = {1 \over S_N(a,b,\beta/2) } \int_{0}^1 du_1 \cdots  \int_{0}^1 du_N \,
\prod_{j=1}^N
\Big |   {d \over dt} \log R(t) \Big |_{t = 0}  \Big |^{h} \\
\times 
\prod_{l=1}^N u_l^a (1 - u_l)^b   \prod_{1 \le j < k \le N} | u_k  - u_j |^\beta,
\end{multline}
where
\begin{equation}\label{1.7a+}
R(t) := \prod_{l=1}^N (t -  u_l)
\end{equation}
and the normalisation, which is referred to as the Selberg integral, has the gamma function
evaluation
\begin{equation}
S_N(\lambda_1,\lambda_2,\lambda) =
 \prod_{j=0}^{N-1} {\Gamma (\lambda_1 + 1 + j\lambda)
\Gamma (\lambda_2 + 1 + j\lambda)\Gamma(1+(j+1)\lambda) \over
\Gamma (\lambda_1 + \lambda_2 + 2 + (N + j-1)\lambda) \Gamma (1 + \lambda )}.
\label{3.2g}
\end{equation}
Our working culminates
in the evaluation of the scaled limit for $h$ a positive integer; see Corollary \ref{C4}.

\subsection{Relationship to a trace distribution for the Cauchy $\beta$-ensemble and consequences}
After recalling (\ref{1.2}) we see that the integrand in (\ref{1.3}) with $h=0$,
under the stereographic projection  mapping $\cot (\theta_j/2) = x_j$ ($j=1,\dots,N)$
 from the unit circle to the real line, transforms according to \cite[Eq.~(3.123)]{Fo10}
\begin{equation}\label{1.6z}
\prod_{l=1}^N   |1 - e^{i \theta_l} |^{2s}   |\Delta_N( \mathbf e^{\mathbf i \mathbold \theta} ) |^\beta d \mathbold \theta =
2^{\alpha_0 - 2s} \prod_{j=1}^N {1 \over (1 + x_j^2)^{\alpha_1}}   |\Delta_N( \mathbf x ) |^\beta \, d \mathbf x,
\end{equation}
where, with $\mathbf x = \{ x_j \}_{j=1}^N$, $\Delta_N(\mathbf x) := \prod_{1 \le j < k \le N} (x_k - x_j)$ and
\begin{equation}\label{1.6v}
\alpha_0 = \beta N (N - 1)/2 + N (1 + 2s) - 2 h, \qquad \alpha_1 = \beta (N - 1)/2 + 1 + s.
\end{equation}
The PDF corresponding to the RHS of (\ref{1.6z})
is known in random matrix theory as the
 symmetric Cauchy
$\beta$-ensemble.

We will see in Proposition \ref{p2.1} below that $ \mathcal{F}_{N,\beta}(s,h)$ has the
interpretation as the moment $\langle |T|^h \rangle$ of the trace statistic 
\begin{equation}\label{TT.1}
T = \sum_{j=1}^N x_j
\end{equation}
for the  symmetric Cauchy
$\beta$-ensemble. In fact in 
deducing Theorem \ref{P1A} we will make use of knowledge of the explicit form 
of the scaled PDF of the trace distribution in this ensemble, $P_{\infty,s}^{(Cy)}(T)$ say.

\begin{theorem}\label{T2}
Let the PDF of the trace statistic (\ref{TT.1}) for the symmetric Cauchy
$\beta$-ensemble corresponding to the RHS of (\ref{1.6z}) be denoted
$P_{N,s}^{(Cy)}(T)$, and define
 \begin{equation}\label{3.1}
{P}_{\infty,s}^{(Cy)} (T)       = \lim_{N \to \infty} N {P}_{N,s}^{(Cy)} (N T).
 \end{equation} 
 For $s \in \mathbb Z_{\ge 0}$ we have
 \begin{equation}\label{3.2x}
 {P}_{\infty,s}^{(Cy)} (T) 
  = {1 \over \pi} {\rm Re} \,    \int_{0}^\infty   e^{-t(1  + i T)} \,  {}_0^{} F_1^{(\beta/ 2)} ( \underline{\: \: \: } \, ; 4 s /\beta ; (4t/\beta)^s ) \, dt.
  \end{equation} 
  Here ${}_0^{} F_1^{(\beta/ 2)}$ is a particular hypergeometric function based on Jack polynomials as defined by
  (\ref{3.40}) below, and where the notation $(4t/\beta)^s$ indicates that each argument $x_i$, $i=1,\dots,s$ is equal
  to $4t/\beta$.
  \end{theorem}
  
  It turns out that knowledge of ${P}_{\infty,s}^{(Cy)} (T) $ is of independent interest.
  To see why, we first make note of the fact that for $\beta = 1,2$ and 4 the RHS of 
  (\ref{1.6z}) is proportional to the eigenvalue PDF on the space Hermitian
  matrices $\{ H \}$ with real $(\beta  = 1)$, complex ($\beta = 2$) and real quaternion
  ($\beta = 4$) chosen according to a distribution with PDF proportional to
  \begin{equation}\label{H1}
  \Big ( \det ( \mathbb I_N + H^2) \Big )^{-\beta(N - 1)/2 - 1 - s};
  \end{equation}  
 see e.g.~\cite[\S 2.5]{Fo10}.
 Matrices from this ensemble have the special property that the sub-matrix obtained by
 deleting $n$ ($n < N$) rows and $n$ columns, also has the distribution implied
 by (\ref{H1}), but with $N$ replaced by $N - n$
 \cite{Hu63, Br95}. This property underlies the Cauchy ensemble being an invariant
 measure --- referred to as a Hua--Pickrell measure ---
 on the space of infinite Hermitian matrices with the action of $U(\infty)$ by
 conjugation; see \cite{BO01} for $\beta = 2$ and \cite{AN19} for $\beta = 1,4$.
 The latter reference shows how a reformulation in terms of the Dixon-Anderson conditional
 PDF (see \cite[Prop.~4.2.1]{Fo10}) extends this to general $\beta > 0$.
 
 On the other hand, a classification of measures which are ergodic under this action --- the
 so-called extremal measures --- have been given in the work of Pickrell 
 \cite{Pi91} and Olshanksi--Vershik \cite{OV96} for $\beta = 2$, and extended to general
 $\beta > 0$ by Assiotis and Najnudel \cite{AN19}, in terms of real parameters
 $( \{ \alpha_i^+ \}_{i=1}^\infty,  \{ \alpha_i^- \}_{i=1}^\infty, \gamma_1, \gamma_2 \}$.
 Moreover an arbitrary invariant measure with respect to the action can, by use of the Fourier
 transform, be decomposed in ergodic components. The task taken up by Borodin and Olshanksi
 \cite{BO01} was to specify the ergodic components, and thus the set of real parameters, in the
 case of the $\beta = 2$ symmetric Cauchy ensemble. This was successfully carried out for the parameters
 $\{ \alpha_i^\pm \}_{i=1}^\infty$. Later Qiu \cite{Qi17} showed that $\gamma_2 = 0$, and
 that $\gamma_1$ has the distribution ${P}_{\infty,s}^{(Cy)} (T)  |_{\beta = 2}$, i.e.~is the
 limiting trace statistic for the $\beta = 2$ symmetric Cauchy ensemble.
 This latter characterisation was extended to general $\beta > 0$ in \cite{AN19}. Hence
 for the parameter $s$ in (\ref{H1}) a non-negative integer, the result of Theorem \ref{T2} has
 the intepretation of specifying the distribution of the real parameter $\gamma_1$ for general
 $\beta > 0$.

\subsection{Overview of the strategy of Winn in relation to $\mathcal F_N(s,h)$, and proposed $\beta$ modification}
The first identity of interest from  \cite{Wi12} gives that
\begin{multline}\label{1.4}
\mathcal{F}_N(s,h) = {2^{N^2 + 2 s N - 2h} \over (2 \pi)^N N!} \int_{\mathbb R^N} \,
\prod_{j=1}^N {1 \over (1 + x_j^2)^{N + s}} | x_1 + \cdots + x_N|^{2h} \Big ( \Delta_N(\mathbf x) \Big )^2 \, d \mathbf x,
\end{multline}
where $ \Delta_N(\mathbf x)$ is as in (\ref{1.6z}) and
$d \mathbf x = dx_1 \cdots dx_N$. This is extended to the general $\beta > 0$ case in \S \ref{S2.1}.
As noted in relation to (\ref{1.6z}),
with $h=0$ the integrand is familiar in random matrix theory as specifying the PDF
for the symmetric Cauchy ensemble with Dyson index $\beta = 2$; see e.g.~\cite[\S 2.5]{Fo10}. 

In the case of $h$ a non-negative integer, the exponential generating function for
the RHS of (\ref{1.4}) is the multiple integral
\begin{equation}\label{1.4a}
\mathcal H_N^{(Cy)} (s;t) :=  \int_{\mathbb R^N} \,
\prod_{j=1}^N {e^{i t x_j}  \over (1 + x_j^2)^{N + s}} \Big ( \Delta_N(\mathbf x) \Big )^2 \, d \mathbf x.
\end{equation}
Equivalently, up to normalisation this quantity is the characteristic function for the distribution of the
trace statistic (\ref{TT.1}).
In  \cite{Wi12} Winn  derives a transformation identity for (\ref{1.4a}), where it was expressed as
$e^{-N |t|}$ times a multiple integral which for $s$ a non-negative integer is polynomial of degree $sN$.
In the general $\beta$ case we don't know how to derive such an identity directly. However, for
$\beta$ even, using the residue calculus, we exhibit in Proposition \ref{P2.2}
that the structure of $ e^{-N |t|}$ times a 
polynomial of degree $sN$ persists. This simple analytic structure in the case
$s$ a non-negative integer 
 enables further calculation.

Winn proceeds
to show that the multiple integral specifying the polynomial can be identified with a certain generalised hypergeometric function
based on Jack polynomials (see \S \ref{S2.2} below for the definitions and notation), showing in particular that
for $s$ a non-negative integer
\begin{equation}\label{1.4b}
\mathcal H_N^{(Cy)}(s;t) \propto    e^{-N |t|}  \, {}_1^{} F_1^{(1)}(-s,-2s;(2 |t| )^N ),
\end{equation}
and furthermore making explicit the proportionality constant.
Using insights into averages over the Cauchy ensemble from the recent work \cite{FR20}
 we show in Proposition \ref{P2.2+}  how to derive (\ref{1.4b}) directly. Most importantly this working extends
 to general $\beta > 0$, allowing for the derivation of the general $\beta > 0$ analogue of (\ref{1.4b}).

 The generalised hypergeometric function $ {}_1^{} F_1^{(1)}$ in (\ref{1.4b}) is a function of $N$ variables.
 By interchanging the role of rows and columns in the underlying sum over partitions which
 defines this function, Winn obtains \cite[Eq.~(6.5)]{Wi12} the transformation identity
 \begin{equation}\label{1.4c}
  {}_1^{} F_1^{(1)}(-s,-2s;(2 |t| )^N ) =  {}_1^{} F_1^{(1)}(-N,2s;(-2 |t| )^s ).
 \end{equation}
 Most significant is that generalised hypergeometric function on the RHS is a function of $s$ variables, not $N$.
 After substituting in (\ref{1.4b}) this functional form underlies Winn's analysis, given in \cite[\S 6.2]{Wi12},
 of the large $N$ limiting form (\ref{1.91}) in the case $\beta = 2$, and with the assumption that $h$ is an
 integer or half integer within the required range (\ref{R}).
 As noted in \cite{Wi12}, the integer case results are the same
 as those derived by Dehaye \cite{De08,De10}, using symmetric function/ representation theoretic ideas
 closely related to the special function in (\ref{1.4b}); for recent extensions of the study of the characteristic
 polynomial for $U(N)$ along these lines see \cite{Ri18, BA20}.
 In our Proposition \ref{P2.7} we give the $\beta$ generalisation of (\ref{1.4c}) relevant to providing
  the $\beta$ generalisation of (\ref{1.4b}).
 
 The recent paper \cite{ABGS20a} places emphasis on the limiting scaled PDF for the trace
 statistic (\ref{TT.1}), as defined by (\ref{3.1}). However the details of the working relied not on the
 generalised hypergeometric formula (\ref{1.4c}) directly, but rather a determinant evaluation 
 in terms of Laguerre polynomials for $ {}_1^{} F_1^{(1)}$ in (\ref{1.4b}) as being proportional to
 \cite{FH94,Wi12}
\begin{equation}\label{1.6}
\det \Big [ L_{N+s - 1 - j - k}^{2s - 1} (2 |t|) \Big ]_{j,k=0,\dots,s-1},
\end{equation}
With $I_\nu(z)$ denoting the modified Bessel function, this gives 
\begin{equation}\label{1.6a}
\lim_{N \to \infty} {1 \over \tilde{C}_N} \mathcal H_N^{(Cy)}(s;t/2N) =
(-1)^{s(s-1)/2} {G(2s+1) \over (G(s+1))^2} \, {e^{-|t|/2} \over |t|^{s^2/2}} \,
\det \Big [ I_{j+k+1}(2 \sqrt{|t|} \Big ]_{j,k=0,\dots,s-1},
\end{equation}
where $ \tilde{C}_N = \tilde{C}_N(s,\beta)$ is chosen so that the LHS is unity for $t=0$. Taking its inverse Fourier
transform gives the limiting distribution implied by (\ref{1.4a}). Associated moments can now be
obtained by a direct evaluation. We remark that other essential uses of (\ref{1.6a}) in the present
context were already known from the earlier work
of Conrey, Rubinstein and Snaith \cite{CRS06}, and used subsequently in
\cite{FW06a,BBBGIIK19,BBBCPRS19,AKW20,ABGS20a}. 
Working with the the finite $N$ formula for the characteristic function based on the $\beta$
generalisation of (\ref{1.4c}) we are able to evaluate the limiting scaled PDF for the trace
 statistic as the inverse Fourier transform of the corresponding  characteristic function,
 which is our Theorem \ref{T2}.

Working with the generalised hypergeometric function
expression for the finite $N$ trace distribution, as implied by taking the inverse Fourier
transformation of the formula for the characteristic function based on (\ref{1.4c}),
a finite summation formula for the moments in the range (\ref{R}) can be obtained.
This is given in Corollary \ref{C3.5}. Its large $N$ form is readily computed, and
implies the result of Theorem \ref{P1A}.

\begin{remark}
For general $s$ the structure of (\ref{1.4a}) is more complex, as can already be seen from the case $N=1$.
Thus it is well known (see the recent work \cite{Ga19} for references and an independent derivation) that 
then (\ref{1.4a}) --- identified up to scaling of $t$ and name of the parameter $s$ as the characteristic
function of Student's $t$-distribution --- is proportional to $|t|^{s+1/2} K_{s+1/2}(|t|)$, where $K_{s+1/2}$ denotes
 the modified Bessel function of the second kind of order $s + 1/2$. This has the analytic structure of  $ e^{- |t|}$ times
 a linear
 combination of a power series in $|t|$ and $|t|^{2s+1}$ times a powers series in $t$. Also, the latter involves
 a logarithm for $s+1/2$ a nonnegative  integer. For $s$ a non-negative integer the
 former power series terminates to a polynomial of degree $s$, and the scalar multiplying the functional form of
 $|t|^{2s+1}$ times a powers series in $t$ vanishes. Further insight into the analytic structure present for continuous $s$ can
 be obtained by revisiting the known \cite{FW01a} Painlev\'e nonlinear differential equation characterising the scaled limit of $\mathcal H_N^{(Cy)} (s;t) $
 evaluated by (\ref{1.6a}).  This same differential equation has recently been shown to hold for continuous $s$ \cite{ABGS20a}.
 Results from \cite{FW06a} can be used to specify the small $t$ boundary condition for continuous $s$.
 This is done in the Appendix, which reveals a so-called Puiseux-type series expansion involving in general an infinite series of
 fractional powers $|t|^{p(2s+1)}$ ($p=1,2,\dots$), each times a power series in $|t|$.
 \end{remark}

\section{Transforming $\mathcal F_{N,\beta}(s,h)$}
\subsection{Generalising (\ref{1.4})}\label{S2.1}
Starting from the definition (\ref{1.3a}), we seek an identity  which generalises (\ref{1.4}).

\begin{proposition}\label{p2.1}
Let $C_{N,\beta}$ be given by (\ref{1.3b}), and make use of the notation $\Delta_N(\mathbf x)$
as specified below (\ref{1.6z}).
We have
\begin{multline}\label{1.6x}
\mathcal F_{N,\beta}(s,h) = {2^{\alpha_0} \over C_{N,\beta}} 
 \int_{-\infty}^\infty dx_1 \cdots  \int_{-\infty}^\infty dx_N \,
\prod_{j=1}^N {1 \over (1 + x_j^2)^{\alpha_1}} | x_1 + \cdots + x_N|^{2h} \Big | \Delta_N(\mathbf x) \Big |^\beta ,
\end{multline}
where
\begin{equation}\label{1.6v}
\alpha_0 = \beta N (N - 1)/2 + N (1 + 2s) - 2 h, \qquad \alpha_1 = \beta (N - 1)/2 + 1 + s.
\end{equation}
\end{proposition}

\begin{proof}
It turns out that in (\ref{1.6x}) the Dyson index $\beta$ does not play an essential role, and so we can follow
the working of \cite[Prop.~3.4]{Wi12} in the case $\beta = 2$. The first step is to note from the
definitions that
$$
{d \over d \theta} \log V_U (\theta) \Big |_{\theta = 0} = - {1 \over 2} \sum_{j=1}^N \cot {\theta_j \over 2}.
$$
Substituting in  (\ref{1.3a})  gives
\begin{equation}\label{1.6y}
\mathcal F_{N,\beta}(s,h) = {1 \over C_{N,\beta}} {1 \over 2^{2h} }\int_0^{2 \pi} d \theta_1 \cdots  \int_0^{2 \pi} d \theta_N \,
\prod_{l=1}^N |1 - e^{i \theta_l} |^{2s} \Big | \sum_{j=1}^N \cot {\theta_j \over 2} \Big |^{2h} |\Delta_N( \mathbf e^{\mathbf i \mathbold \theta} ) |^\beta.
\end{equation}

The second and final step is to use the familiar fact in random matrix theory that under the mapping $x_j = \cot (\theta_j/2)$ ($j=1,\dots,N)$,
which geometrically can be viewed as a stereographic projection from the unit circle to the real line, the integrand in
(\ref{1.6y}) with $h=0$ transforms according to (\ref{1.6z}).
The stated identity (\ref{1.6x}) follows.
\end{proof} 

\begin{remark}
1.~The generalisation of the circular $\beta$-ensemble on the LHS of (\ref{1.6z}) is given the name
circular Jacobi $\beta$-ensemble in random matrix theory; see \cite[\S 3.9]{Fo10}. For a recent study
see \cite{FLT20}.   \\
2.~The Cauchy ensemble average of $\sum_{j=1}^N  |x_j|^{2h}$, rather than $ | x_1 + \cdots + x_N|^{2h}$ as in
(\ref{1.6x}), has been the subject of the recent works \cite{ABGS20} (for $\beta = 2$) and
\cite{FR20} (for $\beta = 1,2$ and 4).  
\end{remark}

\subsection{Generalising (\ref{1.4b})}\label{S2.2}
A generalisation of the PDF corresponding to the RHS of (\ref{1.6z}),
\begin{equation}\label{1.7}
{1 \over \mathcal N_N^{(Cy)}} \prod_{j=1}^N {1 \over (1 + i x_j)^{\beta (N - 1)/2 +1 + s_+} (1 - i x_j)^{\beta (N - 1)/2 +1 + s_-}}   |\Delta_N( \mathbf x ) |^\beta ,
\end{equation}
where $ \mathcal N_N^{(Cy)}=  \mathcal N_N^{(Cy)}(s_+, s_-, \beta)$ is the normalisation, which is well defined for Re$\,(s_+ + s_-) > -1$.
We will refer to this as the nonsymmetric Cauchy
$\beta$-ensemble, which with
$s_+ = s_- = s$ specialises to the  symmetric Cauchy
$\beta$-ensemble corresponding to the RHS of (\ref{1.6z}).
Note that only in the case $s_-$ is equal to complex conjugate of $s_+$ can this be interpreted as a PDF, as otherwise it
may not be real and non-negative. However this does not impact on our working.
 We remark too that in the symmetric case the normalisation is evaluated in
terms of gamma functions \cite[\S 4.7.1]{Fo10} according to
\begin{equation}\label{1.7a}
 \mathcal N_N^{(Cy)} = 2^{-\beta N (N - 1)/2 - 2N s} \pi^N  \,M_N(s,s,\beta/2),
\end{equation}
where
\begin{equation}\label{2.2b}
M_N(a,b,\lambda) = \prod_{j=0}^{N - 1} { \Gamma (\lambda j+a+b+1)
\Gamma(\lambda (j+1)+1) \over
  \Gamma (\lambda j+a+1)\Gamma (\lambda j+b+1) \Gamma (1 + \lambda)}.
\end{equation}  

For the symmetric Cauchy $\beta$-ensemble consider the trace statistic (\ref{TT.1})
and let $P_{N,s}^{(Cy)}(T)$ denote its PDF. From the definitions
\begin{multline}\label{1.7b}
m_{N,s}^{(Cy)}(h) := \int_{-\infty}^\infty |T|^{2h} P_{N,s}^{(Cy)}(T) \, dT   \\  =
{1 \over \mathcal N_N^{(Cy)}}\int_{-\infty}^\infty dx_1 \cdots  \int_{-\infty}^\infty dx_N \,
 | x_1 + \cdots + x_N|^{2h} 
  \prod_{j=1}^N {1 \over (1 + x_j^2)^{\beta (N - 1)/2 +1 + s}}   |\Delta_N( \mathbf x ) |^\beta,
\end{multline}  
which we see from the behaviour of the integrand for $|x_j| \to \infty$ as decaying at a rate
proportional to $|x_j|^{-2(1+s)}$, is well
defined for the parameter range (\ref{R})
in keeping with the condition associated with (\ref{1.1}).
Thus, recalling too  (\ref{1.6x}),
 knowledge of $P_{N,s}^{(Cy)}(T)$ is sufficient to compute  $\mathcal F_{N,\beta}(s,h)$,
 where the average is with respect to the PDF (\ref{1.7}) with $s_+ = s_- = s$.

 More tractable than studying $P_{N,s}(T) $ directly is to study its 
 characteristic function (Fourier transform), $\hat{P}_{N,s}(t)$ say, specified by
\begin{equation}\label{1.7c} 
\hat{P}_{N,s}^{(Cy)}(t) =  \Big \langle  \prod_{j=1}^N e^{i t x_j} \Big \rangle_s^{(Cy)},
\end{equation} 
where the average is with respect to the PDF (\ref{1.7}) with $s_+ = s_- = s$.
In fact it is advantageous to generalise (\ref{1.7c}) so that the average is  with respect to
the nonsymmetric Cauchy weight (\ref{1.7}) with $s_+, s_-$ independent 
and thus to consider
\begin{equation}\label{2.9b} 
\hat{P}_{N,s_+, s_-}^{(Cy)}(t) =  \Big \langle  \prod_{j=1}^N e^{i t x_j} \Big \rangle_{s_+, s_-}^{(Cy)}.
\end{equation} 
Our immediate aim is to extend to the general $\beta > 0$ case the $\beta = 2$ transformation
formula (\ref{1.4b}) due to Winn \cite{Wi12}, involving
hypergeometric
functions based on Jack polynomials; for an account of the latter see \cite[Ch.~12\&13]{Fo10}.

Such hypergeometric functions are functions of several variables, $\{x_j \}_{j=1}^m$ say, as well
as a parameter $\alpha$.
They can be defined by the series
 \begin{equation}\label{3.40}
 _p F_q^{(\alpha)}(a_1,\dots,a_p;b_1,\dots,b_q;x_1,\dots,x_m):=\sum_\kappa \frac{1}{|\kappa|!}\frac{[a_1]^{(\alpha)}_\kappa\dots [a_p]^{(\alpha)}_\kappa }{[b_1]^{(\alpha)}_\kappa
\dots [b_q]^{(\alpha)}_\kappa} 
C_\kappa^{(\alpha)}(x_1,\dots,x_m).
\end{equation}
In this expression the sum is over all partitions $\kappa_1 \ge \kappa_2 \ge \cdots \ge \kappa_m \ge 0$. Conventionally this sum is
performed in order of increasing 
$|\kappa| : = \sum_{j=1}^m \kappa_j$. 
The symbol $[u]_\kappa^{(\alpha)}$ denotes the generalised Pochhammer symbol, and is defined by
 \begin{equation}\label{C0}
 [a]_\kappa^{(\alpha)} = \prod_{j=1}^m \Big ( a - {1 \over \alpha} (j-1) \Big )_{\kappa_j}, \qquad (a)_k = a (a + 1) \cdots (a+k - 1) = {\Gamma(a+k) \over \Gamma(a)}.
 \end{equation}
 The functions $\{ C_\kappa^{(\alpha)}(x_1,\dots,x_m) \}$ are polynomials 
 proportional to the Jack symmetric polynomial (see e.g.~\cite[\S 12.6]{Fo10}). As such they
are homogeneous symmetric polynomials of degree $|\kappa|$. For $m=1$, $C_\kappa^{(\alpha)}(x) = x^{\kappa_1}$,
$ [a]_\kappa^{(\alpha)} = (a)_k$ and (\ref{3.40})
reduces to the classical definition of ${}_qF_p$ in one variable.  As with ${}_qF_p$ in one variable, ${}_qF_p^{(\alpha)}$ exhibits the confluence
property
 \begin{equation}\label{C1}
 \lim_{a_p \to \infty} {}_p F_q^{(\alpha)}(a_1,\dots,a_p;b_1,\dots,b_q; 
 { x_1 \over a_p},\dots, {x_m \over a_p})= {}_{p-1} F_q^{(\alpha)}(a_1,\dots,a_{p-1};b_1,\dots,b_q;x_1,\dots,x_m). 
\end{equation}

It is well known that hypergeometric
functions based on Jack polynomials are naturally related to the Jacobi $\beta$-ensemble from random matrix theory;
see e.g.~\cite[Ch.~13]{Fo10}. The Jacobi $\beta$-ensemble conventionally refers to the eigenvalue PDF supported on
$x_l \in (0,1)$ $(l=1,\dots,N)$ proportional to
 \begin{equation}\label{J1}
 \prod_{l=1}^N x_l^{\lambda_1} (1 - x_l)^{\lambda_2} | \Delta_N( \mathbf x ) |^\beta ,
\end{equation}
where it is required that $\lambda_1, \lambda_2 > -1$ for this to be normalisable.
In the case $N = 1$ (\ref{J1})  reduces to the PDF for the beta distribution B$[\lambda_1+1, \lambda_2+1]$ in probability theory and statistics,
specified by the functional form $x^{\lambda_1} (1 - x)^{\lambda_2}/ B(\lambda_1+1, \lambda_2+1)$, where 
 \begin{equation}\label{J1x}
B(a, b) = {\Gamma(a)  \Gamma(b)   \over   \Gamma(a + b)}
\end{equation}
refers to the Euler beta function in usual notation. It is a classical result that this single variable PDF relates
to the Gauss hypergeometric function --- that is the case $p=2, q=1, m=1$ of (\ref{3.40}) --- via the integral representation
 \begin{equation}\label{J2}
 {1 \over B(\lambda_1+1, \lambda_2+1) }  \int_0^1 x^{\lambda_1} (1 - x)^{\lambda_2} (1 - t x)^{-r} \, dx =
 {}_2^{} F_1 (r, \lambda_1 + 1; \lambda_1 + \lambda_2 + 2; t).
 \end{equation}
Similarly in the multivariable case, the way (\ref{J1}) relates to the hypergeometric functions based on Jack polynomials
is through particular integral representations.

To make use of these known results for the  Jacobi $\beta$-ensemble in the problem of analysing the Cauchy $\beta$-ensemble
quantity (\ref{1.7c}), we require theory from the recent work \cite{FR20} which relates averages in the two ensembles. To present this
theory, with the PDF for the nonsymmetric Cauchy $\beta$-ensemble specified by (\ref{1.7}), denote the ensemble average of $f = f(x_1,\dots,x_N)$ by
$\langle f  \rangle^{(Cy)}_{s_+, s_-}$. Similarly, with the PDF for the Jacobi $\beta$-ensemble specified by the normalised version of
(\ref{J1}), denote the corresponding average of $f$ by $\langle f \rangle^{(J)}_{\lambda_1,\lambda_2}$. For both classes of
averages the dependence on $N$ and $\beta$ is implicit. With this notation, we have from \cite[Prop.~2.1]{FR20} that for
$f$ a multivariable polynomial
 \begin{equation}\label{J3}
  \langle f ( 1 + i x_1,\dots,  1 + i x_N)  \rangle^{(Cy)}_{s_+, s_-} =
  \Big \langle f  \Big (2 y_1,\dots, 2 y_N  \Big ) \Big \rangle^{(J)} \Big |_{\lambda_1 = - \beta (N - 1)/2 - 1 - s_+ \atop  \lambda_2  = - \beta (N - 1)/2 - 1 - s_-} .
 \end{equation}
 An essential point regarding the interpretation of this identity is that both sides are to be understood in the sense of analytic
 continuation in  $s_+, s_-$.
 
 Before considering the general $N$ case, it is instructive to first specialise to $N = 1$. 
 With $f(x) = x^p$, $p \in \mathbb Z_{\ge 0}$, (\ref{J3}) then reads
  \begin{equation}\label{2.19}
  \langle (1 + i x )^p \rangle_{s_+, s_-}^{(Cy)} =  \langle  (2 x)^p \rangle^{(J)}   \Big |_{\lambda_1 =  - 1 - s_+ \atop  \lambda_2  =  - 1 - s_-} .
  \end{equation}
  Note that the LHS is well defined for
    \begin{equation}\label{2.19a}
    0 \le p \le {\rm Re} \, (s_+ + s_-) - 1.
   \end{equation}
    We want to use knowledge of the functional form of the RHS of (\ref{2.19}) to compute the LHS
    in the range (\ref{2.19a}).
    By definition, the exponential generating function of the  moments specifying the RHS of (\ref{2.19}) is given by
   \begin{equation}\label{2.20}
   \langle e^{2t x} \rangle^{(J)} \Big |_{\lambda_1 =  - 1 - s_+ \atop  \lambda_2  =  - 1 - s_-}  =
   {1 \over B(-s_+, - s_-)} \int_0^1 x^{-s_+-1} (1 - x)^{-s_- - 1} e^{2 t x} \, dx.
     \end{equation}
     This is well defined as presented for Re$\,(s_+)$, Re$\,(s_-) < 0$, and defined more generally
     by analytic continuation. The integral on the RHS of (\ref{2.20}) is a confluent limit of (\ref{J2})  (replace $t$ by $-2t/r$ and take
     $r \to \infty$) telling us that
     \begin{equation}\label{2.21}  
     {1 \over B(-s_+,-s_-)} \int_0^1 x^{-s_+-1}(1 - x)^{-s_- - 1} e^{2tx} \, dx =
     {}_1 F_1 (-s_+, - s_+-s_-; 2t).
   \end{equation}
   Equating like powers of $t$, making use of (\ref{3.40}) with $p=q=m=1$ on the RHS
   we can read off the moments, which are given in terms of particular
   Pochhammer symbols as defined in (\ref{C0}).
   Use now of (\ref{2.19}) gives the sought evaluation
    \begin{equation}\label{2.21a}  
   \langle (1 + i x )^p \rangle_{s_+, s_-}^{(Cy)}  =  2^p { (-s_+)_p \over (-s_+ - s_-)_p}.
     \end{equation}
     Note from the definition of the Pochhammer symbols 
     that the RHS is indeed well defined in the parameter range specified by (\ref{2.19a}).
     
     The utility of (\ref{2.21a}) in calculating the $N = 1$ cases of (\ref{1.7c}) and (\ref{2.9b}) shows itself if we restrict to $s \in \mathbb Z_{\ge 0}$ and
     $s_+ \in \mathbb Z_{\ge 0}$ respectively. For definiteness, consider (\ref{2.9b}) with $s_+ \in \mathbb Z_{\ge 0}$, and require too that $t > 0$.
     The significance of restricting to $s_+ \in \mathbb Z_{\ge 0}$ is that then the integrand is analytic in the upper half plane, apart from a pole of order
     $1 + s_+$ at $z = i$. And with $t > 0$ the integrand decays exponentially fast along the arc $R e^{i \theta}$, $0 < \theta < \pi$, $R \to \infty$.
     Closing the contour in the upper half plane, the residue calculus implies
   \begin{equation}\label{2.24}
   \langle e^{t(1 + i x)} \rangle_{s_+, s_-}^{(Cy)} = q_{s_+}(t;s_-),
   \end{equation}      
   for some polynomial $q_{s_+}$ of degree $s_+$ in $t$. We know the moments obtained by 
   expanding the LHS in $t$ are all well defined in the parameter range (\ref{2.19a}). Requiring then that
   $s_+ < {\rm Re} \, (s_+ + s_-)$, the polynomial $q_{s_+}$ is fully determined by the moments.
   Reading from (\ref{2.21a}) we conclude
   \begin{equation}\label{2.25} 
   q_{s_+}(t;s_-) = \sum_{p=0}^{s_+} {(-s_+)_p \over p! ( - s_+ - s_-)_p} (2t)^p.
    \end{equation}  
    
    We seek to extend (\ref{2.24}) and (\ref{2.25}) for $N > 1$ and general $\beta > 0$.  The
    extension of (\ref{2.24}) will be deduced first.
    
    \begin{proposition}\label{P2.2}
Suppose $t > 0$, $s_+ \in \mathbb Z_{\ge 0}$,
Re$\, (s_+ + s_-) \ge 0$
 and that $\beta$ is even. We have
\begin{equation}\label{2.30a}
  \langle  e^{ i t (x_1 + \cdots + x_N) }   \rangle^{(Cy)}_{s_+, s_-} = e^{-Nt}  Q_{s_+ N}(t;s_-),
  \end{equation}  
where $ Q_{s_+ N}(t;s_-)$ is a polynomial in $t$ of degree $s_+N$.
Moreover, with
\begin{equation}\label{2.30b}
{\rm Re} \, (s_-) \ge (N - 1) s_+
 \end{equation} 
 we have
\begin{equation}\label{2.30c} 
Q_{s_+ N}(t;s_-) = \sum_{p=0}^{s_+ N} { t^p \over p!} \Big \langle  \Big ( i \sum_{i=1}^N x_i + N \Big )^p
\Big \rangle_{s_+, s_-}^{(Cy)}.
 \end{equation}  
\end{proposition}

\begin{proof}
In the case $\beta$ even we can expand $  |\Delta_N( \mathbf x ) |^\beta$ as a multivariable
polynomial in $\{x_j\}_{j=1}^N$, homogeneous of degree $\beta N (N - 1)/2$. 
The integrand
is an analytic function in each $x_j$ in the upper half plane, except for a pole of
degree $\alpha_+$, where $\alpha_\pm := \beta (N - 1)/2 + 1 + s_\pm$, at $x_j = i$ for each $j$. Hence, by the residue theorem, the
integral defining the average is proportional to
\begin{equation}\label{r1}
e^{- N t} {\partial^{\alpha_+ - 1} \over \partial x_1^{\alpha_+ - 1} } \cdots  {\partial^{\alpha_+ - 1} \over \partial x_N^{\alpha_+ - 1} } 
\prod_{l=1}^N {e^{i t ( x_l - i)} \over (1 - i x_l)^{\alpha_-} }   (\Delta_N( \mathbf x ) )^\beta  \Big |_{x_1 = \cdots = x_N = i}.
 \end{equation}
We see immediately the structure (\ref{2.30a}), although it remains to determine the
degree of the polynomial.

From the meaning of the factor $(\Delta_N( \mathbf x ))^\beta$, we see that (\ref{r1}) vanishes unless exactly $\beta N (N - 1)/2$ derivative
operations are applied to this factor. As a consequence, a maximum
of $s_+ N $ derivative operations can act of the factor $\prod_{l=1}^N e^{i t ( x_l - i)}$, telling us that the highest
power of $t$ in $Q_{sN}(t)$ is $s_+N$, which is the required result.

To derive (\ref{2.30c}), note that the assumption (\ref{2.30b}) implies each of the averages in (\ref{2.30c}) are well
defined, thus implying that it is a consequence of (\ref{2.30a}), multiplied through by $e^{Nt}$.
\end{proof}

\begin{remark}
The analytic structure of an exponential times a polynomial has been present, and
played an important role, in a number of
recent studies of distributions in random matrix theory \cite{Ku19,FK19,FT19}.
\end{remark}

The task now is to evaluate the averages in (\ref{2.30c}) which determine the polynomial $Q_{s_+ N}(t;s_-) $.

   \begin{proposition}\label{P2.2+}
   With $s_+ \in \mathbb Z_{\ge 0}$ and the condition (\ref{2.30b})   we have
  \begin{equation}\label{2.30d}  
  Q_{s_+ N} (t; s_-) = {}_1^{} F_1^{(2/\beta)} \Big ( -s_+, - (s_+ + s_-); (2t)^N \Big ),
  \end{equation}
  where the notation $(u)^N$ denotes the argument 
   $u$ repeated $N$ times, and ${}_1 F_1^{(2/\beta)}$ denotes the series implied by (\ref{3.40}).
   Moreover, the condition (\ref{2.30b})  can be weakened to 
    \begin{equation}\label{2.30e}  
  {\rm Re} \, (s_+ + s_-) \ge 0, 
 \end{equation}    
 provided we interpret ${}_1 F_1^{(2/\beta)} $ in (\ref{2.30d}) as terminating when the  largest part of the partition $\kappa$ greater than $s_+$.
     
  \end{proposition}

\begin{proof}
A fundamental result in Jack polynomial theory is an integration formula
 associated with Macdonald--Kadell--Kaneko (see \cite{FW08}),
  \begin{eqnarray}\label{15.pse1}
&&{1 \over S_N(\lambda_1,\lambda_2;1/\alpha)} 
\int^1_0 dt_1 \cdots \int^1_0 dt_N \,
\prod_{l=1}^N t_l^{\lambda_1} (1-t_l)^{\lambda_2}
C_{\kappa}^{(\alpha)} (t_1,\ldots,t_N) \prod_{j<k}\left|t_j-t_k\right|^{2/\alpha}
\nonumber \\
&& \qquad \qquad =  C_{\kappa}^{(\alpha)} ((1)^N) \frac{[\lambda_1 +(N-1)/\alpha +1]^{(\alpha)}_{\kappa} }
{[\lambda_1 + \lambda_2 +2(N-1)/\alpha +2]^{(\alpha)}_{\kappa}}.
\end{eqnarray}
Here the normalisation $S_N$ is the Selberg integral, which as already remarked has
the gamma function evaluation (\ref{3.2g}). The polynomials $\{ C_{\kappa}^{(\alpha)} \}$ are the Jack
polynomials in the same normalisation as in (\ref{3.40}). The value of $ C_{\kappa}^{(\alpha)} ((1)^N) $ is
specified in Remark \ref{R3.4+} below.

To be able to make use of (\ref{15.pse1}), it needs to be possible to expand the polynomial 
$\Big ( i \sum_{i=1}^N x_i + N \Big )^p$ in (\ref{2.30c})
in terms of $\{ C_{\kappa}^{(\alpha)} \}$. In fact we have the simple formula (see \cite[eq.~(12.135)]{Fo10})
 \begin{equation}\label{J3Z}
 (x_1 + \cdots + x_N)^p = \sum_{|\kappa| = p} C_\kappa^{(\alpha)}(x_1,\dots,x_N).
\end{equation}  
It follows 
 \begin{equation}\label{H3Z}
 \Big \langle  (x_1 + \cdots + x_N)^p  \Big \rangle^{(J)}_{\lambda_1,\lambda_2}
 \bigg |_{\lambda_1 =    - \beta (N - 1)/2 - 1 - s_+ \atop  \lambda_2 = - \beta (N - 1)/2 - 1 - s_-}  
    =  \sum_{|\kappa| = p}   C_{\kappa}^{(2/\beta)} ((1)^N) \frac{[- s_+ ]^{(2/\beta)}_{\kappa} }
{[-  s_+ - s_- ]^{(2/\beta)}_{\kappa}},
\end{equation}  
where the RHS is be understood as providing an analytic continuation in $s_+, s_-$ of the LHS. Upon application of
 (\ref{J3}) we obtain for the moments in (\ref{2.30c})
  \begin{equation}\label{H3Z1}
  \Big \langle  \Big ( i \sum_{i=1}^N x_i + N \Big )^p
\Big \rangle_{s_+, s_-}^{(Cy)} =  2^{|\kappa|} \sum_{|\kappa| = p} 
C_{\kappa}^{(2/\beta)} ((1)^N) \frac{[- s_+ ]^{(2/\beta)}_{\kappa} }
{[-  s_+ - s_- ]^{(2/\beta)}_{\kappa}}.
\end{equation}  
Note that the denominator in (\ref{H3Z1}) is well defined by the assumption (\ref{2.30b})
for all partitions with $\kappa_1 \le s_+$ since $p \le s_+N$.
Substituting in (\ref{2.30c}) and comparing with the series definition of 
${}_1 F_1^{(2/\beta)}$ as implied by (\ref{3.40}) gives (\ref{2.30d}).

With the condition (\ref{2.30b})
    \begin{equation}\label{2.30f}  
 {}_1^{} F_1^{(2/\beta)} \Big ( -s_+, - (s_+ + s_-); (2t)^N \Big ) =     \sum_{p=0}^{s_+ N} 
\sum_{|\kappa| = p \atop \kappa_1 \le s_+} 
{C_{\kappa}^{(2/\beta)} ((2t)^N) \over | \kappa |! }\frac{[- s_+ ]^{(2/\beta)}_{\kappa} }
{[-  s_+ - s_- ]^{(2/\beta)}_{\kappa}}.
\end{equation}
Substituting (\ref{2.30d}) in  (\ref{2.30a}), which we know is valid for $s_+ \in \mathbb Z_{\ge 0}$,
Re$\, (s_+ + s_-) \ge 0$, we see that (\ref{2.30f}) remains valid while it is an analytic
function of $s_-$, which is for the range (\ref{2.30e}).
  \end{proof}
  
  \begin{remark}\label{R3.4+}
  The value of $C_\kappa^{(\beta/2)}  ( (1)^s )$ is known explicitly from Jack polynomial
  theory (see e.g.~\cite[Ch.~12]{Fo10}). It involves the diagram of the partition $\kappa$;
  see e.g.~\cite[Def.~10.1.3]{Fo10}.
  On this, recall the diagram records each nonzero part $\kappa_i$ as row
   $i$ of $\kappa_i$ boxes, drawn flush left starting from the first column. Consider a square $(i,j)$ in
   the diagram. The arm length $a(i,j)$ is defined as the number of boxes in the row to the
   right and the co-arm length $a'(i,j)$ is the number in the row to the left. The leg length $\ell(i,j)$ 
   is the number of boxes in the column below and the co-leg length $\ell'(i,j)$  is the number in the column
   and above the square.
  In terms of this notation we have
   \begin{equation}\label{2.6+} 
 C_\kappa^{(\alpha)}((1)^n) =  \alpha^{|\kappa|} |\kappa|!   {b_\kappa \over d_\kappa' h_\kappa},
  \end{equation}  
  where $b_\kappa, d_\kappa' , h_\kappa$ are specified in terms of the diagram of $\kappa$ according to
   \begin{equation}\label{2.6a+} 
 b_\kappa = \prod_{(i,j) \in \kappa} \Big ( \alpha a'(i,j) + n -   \ell'(i,j) \Big ) = \alpha^{|\kappa|} [n/\alpha]_\kappa^{(\alpha)}, \quad
d_\kappa'  =   \prod_{(i,j) \in \kappa} \Big ( \alpha (a(i,j) + 1) +   \ell(i,j) \Big )
  \end{equation}  
  and
  \begin{equation}\label{2.6b+} 
 h_\kappa  =   \prod_{(i,j) \in \kappa} \Big ( \alpha a(i,j)  +   \ell(i,j) + 1\Big ).    
  \end{equation}  
  The simplest case is when $n=1$, and we have
    \begin{equation}\label{2.6c+} 
 C_\kappa^{(\alpha)}((1)^n) \Big |_{n=1} =  1
  \end{equation}
  for all $\kappa = (\kappa_1,0,\dots,0)$ in keeping with $ C_\kappa^{(\alpha)}(x) = x^{\kappa_1}$,
  as already noted below (\ref{C0}). 
  \end{remark}
  
  \begin{corollary}\label{C1}
  Suppose $t > 0$, $s \in \mathbb Z_{\ge 0}$.
 For all $\beta > 0$ we have
\begin{align}\label{2.30a5}
  \langle  e^{ i t (x_1 + \cdots + x_N) }   \rangle^{(Cy)}_{s} & = e^{-Nt} \,
 \lim_{\epsilon \to 0}  {}_1^{} F_1^{(2/\beta)} \Big ( -s, - 2 s - \epsilon; (2t)^N \Big ) \nonumber \\
 & =  e^{-Nt} \sum_{p=0}^{s N} 
\sum_{|\kappa| = p \atop \kappa_1 \le s} 
{(2t)^{|\kappa|} C_{\kappa}^{(2/\beta)} ((1)^N) \over | \kappa |! }\frac{[- s ]^{(2/\beta)}_{\kappa} }
{[-  2 s ]^{(2/\beta)}_{\kappa}}.
\end{align} 
Moreover, replacing $t$ by $|t|$ on the RHS allows the condition $t > 0$ to be replaced
by $t \in \mathbb R$.
\end{corollary}

\begin{proof}
Substituting (\ref{2.30d}) in (\ref{2.30a}) and recalling (\ref{2.30f}) gives (\ref{2.30a5}) valid for $\beta$ even.
To remove the latter restriction, note from (\ref{3.40})), (\ref{C0})
and (\ref{2.6+}),
that each term on the RHS of (\ref{2.30a5})
 is a bounded analytic function of $\beta$  for Re$ \, (\beta) > 0$.
 On the LHS the function being averaged --- a complex exponential ---
 has modulus equal to 1, and so it is immediate that
this side is similarly a bounded  analytic function of $\beta$  for Re$ \, (\beta) > 0$. Application of Carlson's
theorem (see e.g.~\cite[Prop. 4.1.4]{Fo10}), as is familiar in the theory of the Selberg integral \cite{Se44,Dy62},
then allows the validity of (\ref{2.30a5}) for general $\beta > 0$  to be concluded from knowledge of
its validity for positive even $\beta$.

The fact that the LHS in an even function of $t$ from its definition (thus the PDF proportional to the
RHS of (\ref{1.6z}) for the
symmetric Cauchy $\beta$-ensemble is unchanged by the mapping  $x_j \mapsto - x_j$ for each $j=1,\dots,N$),
and that both sides equal unity for $t=0$ show that replacing $t$ by 
$|t|$ in the RHS of (\ref{2.30a5}) gives an equation valid for all $t \in \mathbb R$.
\end{proof}

    \subsection{Generalising (\ref{1.6})}\label{S2.4}
The generalised hypergeometric function on the RHS of (\ref{2.30a5}),
which is based on Jack polynomials in $N$ variables,
permits a transformation
to involve a hypergeometric function based on
Jack polynomials in $s$. The advantage of this latter form is that is well
suited to taking the scaled large $N$ limit.

\begin{proposition}\label{P2.7}
Let $s  \in \mathbb Z_{\ge 0}$. We have
  \begin{equation}\label{L1z}
 \lim_{\epsilon \to 0}  {}_1^{} F_1^{(2/\beta)} \Big ( -s, - 2 s - \epsilon; (2t)^N \Big )  =
  {}_1^{} F_1^{(\beta/ 2)} ( - N; 4 s /\beta ; (-4t/\beta)^s ).
 \end{equation}
 Consequently
   \begin{equation}\label{L1zg}
  \hat{P}_{N,s}^{(Cy)}(t)  = e^{-N |t|} \,     {}_1^{} F_1^{(\beta/ 2)} ( - N;  4 s  /\beta ; (-4|t|/\beta)^{s} )
   \end{equation}
   and
   \begin{equation}\label{L1z+}
 \hat{P}_{\infty,s}^{(Cy)}(t) := \lim_{N \to \infty}   \hat{P}_{N,s}^{(Cy)}(t/N)=  e^{-|t|} \,  {}_0^{} F_1^{(\beta/ 2)} ( \underline{\: \: \: } \, ;  4 s /\beta ; (4|t|/\beta)^s ).
 \end{equation} 
 \end{proposition}  
 
   \begin{proof}
For $s \in \mathbb Z_{\ge 0}$ and general $s_1 \notin
 \mathbb Z_{\ge 0}$, analogous to the second line in
(\ref{2.30a5}) we have 
\begin{equation}\label{S1}
 {}_1^{} F_1^{(2/\beta)} \Big ( -s, -  s_1; (2t)^N \Big )  =  \sum_{p=0}^{s N} 
\sum_{|\kappa| = p \atop \kappa_1 \le s} 
{(2t)^{|\kappa|} C_{\kappa}^{(2/\beta)} ((1)^N) \over | \kappa |! }\frac{[- s ]^{(2/\beta)}_{\kappa} }
{[-   s_1 ]^{(2/\beta)}_{\kappa}}.
\end{equation}
Hence in the diagram picture of a partition  the sum of $\kappa$
is over partitions with no more than $N$ rows and $s$ columns. Now replace each $\kappa$ by its
conjugate $\kappa'$, which by definition interchange rows and columns. The sum now is over partitions
with diagram consisting of no more that $s$ rows and $N$ columns. Moreover, we have the general
relations between quantities relevant to the summand on the LHS of (\ref{S1}),
 \begin{equation}\label{L1yA}
 [a]_\kappa^{(\alpha)} = (- \alpha)^{-|\kappa|} [-\alpha a]_{\kappa'}^{(1/\alpha)}
 \end{equation}
 and 
  \begin{equation}\label{L1yB+}
  d'_{\kappa'} = \alpha^{|\kappa|} h_\kappa \Big |_{\alpha \mapsto 1/\alpha}, \qquad
  h_{\kappa'} =  \alpha^{|\kappa|} d'_{\kappa}  \Big |_{\alpha \mapsto 1/\alpha},
 \end{equation}
 where the relevance of (\ref{L1yB+}) is seen from (\ref{2.6+}). Thus we obtain the
 transformation identity
   \begin{align}\label{L1yB}
  {}_1^{} F_1^{(2/\beta)} \Big ( -s, -  s_1; (2t)^N \Big ) & =   
 \sum_{p=0}^{s N} 
\sum_{|\kappa| = p \atop \kappa_1 \le N} 
{(-4t/\beta)^{|\kappa|} C_{\kappa}^{(\beta/2)} ((1)^s) \over | \kappa |! }\frac{[-N]^{(\beta/2)}_{\kappa}}
{[ 2   s_1/\beta ]^{(\beta/2)}_{\kappa}}  \nonumber \\
& =  {}_1^{} F_1^{(\beta/2)} \Big ( -N, 2 s_1/\beta; (-4t/\beta)^s \Big ),
\end{align} 
 from which (\ref{L1z}) follows.
 Substituting     (\ref{L1z})  in the first line of (\ref{2.30a5}) gives (\ref{L1zg}).
 Now replacing $t$ by $t/N$ and taking the limit making use of (\ref{C1}) on the RHS of (\ref{L1z}) gives (\ref{L1z+}).
   
\end{proof}

\begin{remark}
1.~Suppose $s=1$. The hypergeometric functions based on Jack polynomials in 
 Proposition   \ref{P2.7} all reduce to their classical counterparts. For example,
 from (\ref{L1zg})
  \begin{equation}\label{sPt}
    \hat{P}_{N,s}^{(Cy)}(t) \Big |_{s=1}   = e^{-N |t|} \,     {}_1^{} F_1( - N; 4  /\beta ; -4t/\beta ).
 \end{equation}     
   A check is to take the limit $\beta \to 0$, since it follows from the definitions that
 for general $s$ such that the LHS is well defined we must have
  \begin{equation}\label{tPt}   
    \hat{P}_{N,s}^{(Cy)}(t) \Big |_{\beta =0}  =  \Big (      \hat{P}_{N,s}^{(Cy)}(t) \Big |_{N=1}  \Big )^N.
 \end{equation} 
Since
 for $\beta \to 0$ the RHS of (\ref{sPt}) reduces to
 $$
  e^{-N |t|} \,     {}_1 F_0( - N; \underline{\: \: \:} \, ; -t ) = \Big ( {e^{- |t|} ( 1 + |t| )} \Big )^N,
  $$
  consistency with (\ref{tPt}) is indeed observed.   \\
  2.~In words, the definition (\ref{1.7c}) of $\hat{P}_{N,s}(t)$ is the characteristic function for the
distribution of a linear statistic $\sum_{j=1}^N x_j$ in the symmetric Cauchy $\beta$-ensemble, or equivalently from the
  working of Proposition \ref{p2.1}, of the linear statistic $\sum_{j=1}^N \cot \theta_j/2$ for the circular
  Jacobi $\beta$-ensemble. From the latter, the linear statistic is thus singular, as it diverges for 
  eigenvalues at $\theta = 0$. As such it exhibits anomalous large $N$ behaviour relative to the characteristic
  function for a smooth
  linear statistic $A = \sum_{j=1}^N a(x_j)$. Denoting the characteristic function by $\hat{P}_N(t;A)$,
  in the smooth case (see e.g.~\cite{PS11})
    \begin{equation}\label{2.52a}
  \lim_{N \to \infty}  \hat{P}_N^{(Cy)}(t/N;A) = e^{i t \mu_A}, \qquad \mu_A = \lim_{N \to \infty} {1 \over N} \langle A \rangle.
   \end{equation}     
  Hence the scaled moments beyond the first are all trivial, just being powers of the mean. For the present singular statistic, the
  scaled (even integer) moments are no longer simply related, as we see
 from (\ref{L1z+}). 

%
  
  \end{remark}
  
  \section{Proof of Theorems \ref{P1A} and \ref{T2}}\label{S3}
  \subsection{Proof of Theorem \ref{T2}}
  Writing the scaled PDF $ N {P}_{N,s}^{(Cy)} (N T)$ in (\ref{3.1}) as the
  inverse Fourier transform of $  \hat{P}_{N,s}^{(Cy)} (t/N)$ we see that
  \begin{equation}\label{3.1a} 
  {P}_{\infty,s}^{(Cy)} (T)  =  {1 \over 2 \pi}  \lim_{N \to \infty} 
  \int_{-\infty}^\infty   \hat{P}_{N,s}^{(Cy)} (t/N)   e^{-i t T} \, dt.
   \end{equation} 
   Application of Levy's continuity theorem permits the limit to be interchanged with
   the integral. Evaluating the limit inside the integral using (\ref{L1z+}) gives
    (\ref{3.2x}). Note that the existence of the limit in (\ref{3.1})
   for general $s \ge 0$ follows from results in \cite{AN19} on the generalised Hua-Pickrell
   measure; recall the discussion below Theorem \ref{T2}.

   \subsection{Evaluating the integral   (\ref{3.2x}) as a sum}
  Upon integrating (\ref{3.2x}) term-by-term particularly simple functional forms result in the cases $s=0$ and $s=1$.
  
  \begin{proposition}\label{P2.8}
  We have
    \begin{equation}\label{3.3a} 
  {P}_{\infty,s}^{(Cy)} (T) \Big |_{s=0} = {1 \over \pi} {1 \over 1 + T^2}
   \end{equation} 
   and
 \begin{equation}\label{3.3b} 
  {P}_{\infty,s}^{(Cy)} (T) \Big |_{s=1} = {1 \over \pi}  \, {\rm Re} \, {1 \over 1 + i T} \, {}_1 F_1\Big (1;4/\beta;  {4/\beta \over 1 + i T} \Big ).
   \end{equation}       
   \end{proposition} 
   
   \begin{proof}
   The case $s=0$ is immediate since the function $ {}_0^{} F_1^{(\beta/ 2)} $ in (\ref{3.2x}) is identically equal to one. The fact that all the
   terms in ${}_0^{} F_1^{(\beta/ 2)} $ are positive for $s \ge 1$, and that the integral is absolutely convergent (this is immediate since the
   large $t$ leading behaviour of ${}_0^{} F_1^{(\beta/ 2)} $ is of order $e^{4 s (t/\beta)^{1/2}}$, as follows from the integral form
   \cite[Eq.~(13.27)]{Fo10}; see in particular the discussion above   \cite[Eq.~(13.52)]{Fo10}) justifies the interchange of the sum and
   integral for $s \in \mathbb Z^+$, and in particular for $s=1$ when ${}_0^{} F_1^{(\beta/ 2)} $ reduces to its classical counterpart.
   \end{proof}
   
   \begin{remark}
   1.~The fact that the limiting distribution in the case $s=0$ is the Cauchy distribution for all $\beta > 0$ can already be
   seen from (\ref{2.30a5}). Thus with $s=0$ this result tells us that $ \hat{P}_{N,s}^{(Cy)} (t/N) = e^{-|t|}$, implying $\sum_{j=1}^N x_j/N$
   has a Cauchy distribution independent of both $\beta$ and $N$. \\
   2.~For the classical values $\beta = 1,2$ and 4 (at least), the RHS of (\ref{3.3b}) can be simplified.
   Thus we have
  \begin{align*}
   {P}_{\infty,s}^{(Cy)} (T) \Big |_{s=1,\beta = 1} & =   {-39 + 3 T^2 \over 32 \pi} + {3  \over  32 \pi } e^{4 \over 1 + T^2} \bigg (
   (1 - T^2) \cos \Big ( {4 T \over 1 + T^2} \Big ) + 2 T  \sin \Big ( {4 T \over 1 + T^2} \Big )   \bigg ) \\
     {P}_{\infty,s}^{(Cy)} (T) \Big |_{s=1,\beta = 2} & = {1 \over 2 \pi} \bigg ( -1 +  e^{2 \over 1 + T^2}    \cos \Big ( {2 T \over 1 + T^2} \Big )  \bigg ) \\
    {P}_{\infty,s}^{(Cy)} (T) \Big |_{s=1,\beta = 4} & =    {  e^{1 \over 1 + T^2} \over \pi (1 + T^2)}   \bigg (   \cos \Big ( { T \over 1 + T^2} \Big ) - T
    \sin \Big ( { T \over 1 + T^2} \Big )    \bigg ).
   \end{align*}
   In the case $\beta = 2$ this agrees with a result reported in \cite{ABGS20a}.
   As implied by (\ref{3.3b}), independent of $\beta$, the leading large $T$ form exhibited by each of the
   above is
    \begin{equation}\label{3.3c} 
      {P}_{\infty,s}^{(Cy)} (T) \Big |_{s=1} = {1 \over (1 + T^2)^2} \Big ( 1 + {\rm O} \Big ( {1 \over 1 + T^2} \Big ) \Big ).
  \end{equation}

   \end{remark} 
   
   Term-by-term integration in (\ref{3.2x}) also evaluates $ {P}_{\infty,s}^{(Cy)} (T) $ for general integers $s \ge 2$,
   however the resulting expression cannot be recognised as a stand alone special function.
 
  \begin{proposition}\label{P2.9}
Let $\kappa = (\kappa_1,\dots, \kappa_s)$ be a partition of $s$ parts.  For all integers $s \ge 1$ we have
    \begin{equation}\label{3.4a}  
  {P}_{\infty,s}^{(Cy)} (T)  =  {1 \over \pi}  \, {\rm Re} \, {1 \over 1 + i T} \,  \sum_\kappa {1 \over  [4s/\beta]_\kappa^{(\beta/2)}}  \Big ( {4 \over \beta (1 + i T)} \Big )^{|\kappa|}
  C_\kappa^{(\beta/2)}  ( (1)^s ).
  \end{equation}
  \end{proposition}
  
    \begin{proof}
    Justification of the interchange of the integral with the sum has already been given in the proof of
    Proposition \ref{P2.8} in the case $s=1$.
    \end{proof}

  \begin{remark}\label{R3.4}
  Taking the inverse Fourier transform of (\ref{L1zg}) and integrating term-by-term
  (the validity of this is immediate since the sum is finite) shows that for finite $N$
  \begin{equation}\label{3.4aF}  
  {P}_{N,s}^{(Cy)} (T)  =  {1 \over \pi}  \, {\rm Re} \, {1 \over N + i T} \,  \sum_\kappa {[-N]_\kappa^{(\beta/2)} \over  [4s/\beta]_\kappa^{(\beta/2)}}  \Big ( -{4 \over \beta (N + i T)} \Big )^{|\kappa|}
  C_\kappa^{(\beta/2)}  ( (1)^s ).
  \end{equation} 
\end{remark}

\subsection{Proof of Theorem \ref{P1A}}

Define the moments $m_{N,s}^{(Cy)}(h)$ associated with $P_{N,s}^{(Cy)}(T)$ by (\ref{1.7b}).
Substitution of (\ref{3.4a}) together with knowledge of the
  integral evaluation \cite[Entry 3.194.3]{GR07}
     \begin{equation}\label{3.4b} 
     \int_0^\infty {x^{2 h} \over (1 - i x)^{k+1}} \, dx = - i \pi e^{i \pi h} {(-2h)_k \over k! \sin 2 \pi h}, \qquad k > 0, \, -1/2 < {\rm Re} \, h < k/2,
 \end{equation} 
 provides for the evaluation of these moments as a finite sum.    
   
   \begin{corollary}\label{C3.5}
Let $\kappa = (\kappa_1,\dots, \kappa_s)$ be a partition of $s$ parts.    For the parameter range (\ref{R}), and with $s \in \mathbb Z^+$, we have
  \begin{equation}\label{3.4d}   
 m_{N,s}^{(Cy)}(h) =   {N^{2h} \over \cos \pi h} \sum_{\kappa}  {[-N]_\kappa^{(\beta/2)} \over N^{|\kappa|}}
 {(-2 h)_{|\kappa|} \over | \kappa |! [4 s/ \beta]_\kappa^{(\beta/2)} }
 \Big ( {4 \over \beta} \Big )^{|\kappa|} C_\kappa ^{(\beta/2)}  ( (1)^s ).
  \end{equation}
  This expression has a removable singularity at $h$ half an odd positive integer, and is to be understood in the sense
  of analytic continuation at these points.
  \end{corollary}
  
  \begin{proof}
  Suppose first that $-1/2 < {\rm Re} \, h < 1/2$. Since each term in (\ref{3.4a}) decays at least as fast as $1/|T|^2$ as 
  $|T| \to \infty$, and furthermore the sum in finite, we can interchange the order of the integral and the sum, and
  the formula (\ref{3.4d}) results. This formula extends to the full range (\ref{R}) by analytic continuation.
  Thus the analyticity of $ m_{\infty,s}^{(Cy)}(h)$ in this range follows from the discussion below (\ref{1.7b}), while the
  RHS is in fact analyic throughout all of ${\rm Re} \, (h) > - 1/2$, with the points $h \in \mathbb Z_{\ge 0}$ being
  removable singularities.
   \end{proof}
  
We are now in a position to prove  Theorem \ref{P1A}. According to (\ref{1.6x}) and (\ref{1.7b}), using the notation therein,
we have
  \begin{equation}\label{7.7x}
  {\mathcal F}_{N,\beta}^{(Cy)} (s,h) = 2^{\alpha_0}  {\mathcal N_N^{(Cy)} \over C_{N,\beta}} m_{N,s}^{(Cy)}.
  \end{equation}
  In the case $s$ a non-negative integer, we can simplify
    \begin{equation}\label{7.7y}
    2^{\alpha_0}  {\mathcal N_N^{(Cy)} \over C_{N,\beta}}  = {1 \over 2^{2h} } 
    \prod_{j=1}^s {\Gamma( 2 (s + j)/\beta + N) \over \Gamma(2 j/\beta + N)} {\Gamma(2j/\beta) \over \Gamma(2(s+j)/\beta)}.
    \end{equation} 
For large $N$ the product over the ratio of the gamma functions involving $N$ has the   leading form $N^{2s^2/\beta}$, as
follows from the standard fact that for any fixed $u$ the leading large $N$ form of
$\Gamma(N+u)/\Gamma(N)$ is equal to $N^u$. We use
this in (\ref{7.7y}) and substitute the resulting expression in (\ref{7.7x}), telling us that
  \begin{equation}\label{7.7z}
  \lim_{N \to \infty} {1 \over N^{2s^2/\beta + 2h}}   {\mathcal F}_{N,\beta}^{(Cy)} (s,h) = {1 \over 2^{2h} }   \prod_{j=1}^s {\Gamma(2j/\beta) \over \Gamma(2(s+j)/\beta)}
   \lim_{N \to \infty} {1 \over N^{2h}}   m_{N,s}^{(Cy)}(h).
   \end{equation}

In the sum of (\ref{3.4d}) for $ m_{N,s}^{(Cy)}(h)$, from the definition (\ref{C0}) we see that we have the bound
\begin{equation}\label{3.12a}
| [-N]_\kappa^{(\beta/2)}/ N^{|\kappa|} | \le r_1^{|\kappa|}
\end{equation}
for some $r_1 > 1$ independent of $N$. In addition, as a consequence of (\ref{J3Z}) and the fact that $C_\kappa^{(\alpha)}((1)^s)$
is non-negative, we have the bound
\begin{equation}\label{3.12b}
|  C_\kappa^{(\alpha)}((1)^s) | \le s^{|\kappa|}.
\end{equation}
Substituting (\ref{3.12a}) and  (\ref{3.12b}) in the sum in (\ref{3.4d}) shows that the absolute value of the
latter is bounded by the convergent sum (since for large $| \kappa |$ the leading asymptotics of $(-2h)_{|\kappa|}| /|\kappa|!$
is $|\kappa|^{-2h - 1}/\Gamma(-2h)$, the definition
(\ref{C0}) and Stirling's formula shows that the decay of the summand is to leading order
 like the reciprocal of a gamma function in each part $\kappa_i$)
\begin{equation}\label{3.12c}
\sum_{\kappa} { |(-2h)_{|\kappa|}| \over | \kappa|! [4s/\beta]_\kappa^{(\beta/2)} } r_2^{|\kappa|}
\end{equation}
for some $r_2 > 0$ and independent of $N$. Consequently, by dominated convergence, the $N \to \infty$ limit of the sum
in (\ref{3.4d})  can be taken term-by-term, and we obtain the limit formula
\begin{equation}\label{3.12d}
   \lim_{N \to \infty} {1 \over N^{2h}}   m_{N,s}^{(Cy)}(h) = 
 {1 \over \cos \pi h }   \sum_{\kappa} 
 {(-2 h)_{|\kappa|} \over | \kappa |! [4 s/ \beta]_\kappa^{(\beta/2)} }
 \Big ( {4 \over \beta} \Big )^{|\kappa|} C_\kappa ^{(\beta/2)}  ( (1)^s ).
 \end{equation}
 Substituting this in (\ref{7.7z}) gives  (\ref{1.91}).

Note that this working also shows that in the range (\ref{R}),
\begin{equation}\label{3.12e}
m_{\infty,s}^{(Cy)}(h) :=   \int_{-\infty}^\infty   |T|^{2h} P_{\infty,s}^{(Cy)} =
\lim_{N \to \infty} {1 \over N^{2h}}   m_{N,s}^{(Cy)}(h),
\end{equation}
as is suggested by the scaling in (\ref{3.1}), and thus has the evaluation (\ref{3.12b}).

\begin{remark}\label{R3.6}
For a given $h \in \mathbb Z_{\ge 0}$, the sum in (\ref{3.12d}) terminates. The explicit evaluation of $C_\kappa^{(\beta/2)}((1)^s)$
implied by (\ref{2.6+}) then tells us that the RHS of (\ref{3.12d}) is a rational function of $s$. For example, with $h=1$, we compute
\begin{equation}\label{3.19x}
m_{\infty,s}^{(Cy)}(h)  \Big |_{h=1} = {1 \over (2 s - 1) (4s/\beta + 1) },
\end{equation}
where use has been made of (\ref{3.12e}) on the LHS. The derivation of (\ref{3.12d}) generally required $s$ to be a non-negative integer such that
${\rm Re} \, h < s + 1/2$, and thus in (\ref{3.19x}) that $s \in \mathbb Z^+$. Yet (\ref{3.19x}) is well defined for all real $s > 1/2$, suggesting
that it in fact holds in that domain. For $\beta = 2$, this statement, and its generalisation for all $h \in \mathbb Z_{\ge 0}$, has been proved true
in the recent work \cite{AKW20}.

For general $\beta > 0$, we can also establish the validity of the RHS of (\ref{3.12d}) for a fixed $h \in \mathbb Z^+$ and continuous $s > h - 1/2$.
We simply observe that in the evaluation formula 
 (\ref{3.4d})   for
 $m_{N,s}^{(Cy)}(h) $, both the LHS (by considering its definition (\ref{1.7b}) and the RHS (using the fact that it is a finite sum, and
 inspecting the analytic form of each term as a function of $s$) that each are well defined in the complex plane
 for ${\rm Re} \, s > h - 1/2$, and moreover are bounded in this half space. Application of Carlson's theorem as used in the
 proof of Corollary \ref{C1} then tells us that the evaluation formula  (\ref{3.4d}) is, with  $h \in \mathbb Z^+$ fixed, value for
 continuous $s > h - 1/2$. Being a finite sum with number of terms independent of $N$, the scaled limit required by (\ref{3.12e}) can be taken term by term, giving
 the formula as implied by (\ref{3.12d}) in the setting $h \in \mathbb Z^+$ fixed, but now valid for  continuous $s > h - 1/2$. 
 \end{remark}

\section{The average (\ref{1.7+})}\label{S4}
Up to normalisation, (\ref{J1}) is the PDF specifying the Jacobi $\beta$-ensemble. Denote an
 average with respect to this ensemble by $\langle \cdot \rangle_{(\lambda_1,\lambda_2,\beta)}^{(J)}$. We see that $\mathcal H_N^{(J)}(h)$
 as defined by  (\ref{1.7+}) can be written as a Jacobi $\beta$-ensemble average according to
 \begin{equation}\label{Se2}
\mathcal H_N^{(J)}(h) =  \Big \langle \Big ( \sum_{j=1}^N {1 \over u_j} \Big )^{h} 
\Big \rangle_{(a ,b,\beta)}^{(J)}.
  \end{equation} 
  This shows that $\mathcal H_N^{(J)}(h)$  is proportional to the $h$-th moment of the singular statistic $\sum_{j=1}^N 1/ u_j$
  for the Jacobi $\beta$-ensemble.
  
  A variant is to change variables $x_l \mapsto {1 \over 2} ( 1 - \cos \theta_l )$ in (\ref{J1}) to obtain the density 
   \begin{equation}\label{Se3}
 {1 \over \tilde{S}_N(a,b,\beta/2)} \prod_{l=1}^N (1 - \cos \theta_l)^{a  + 1/2}  (1 +   \cos \theta_l)^{b+1/2}
  \prod_{1 \le j < k \le N} | \cos \theta_k - \cos \theta_j |^\beta,
 \end{equation} 
 for appropriate normalisation $\tilde{S}(a,b,\beta)$ simply related to the Selberg integral.
 We have that
 \begin{equation}\label{Se4}
\mathcal H_N^{(J)}(h) \propto 2^{h}   \Big \langle \Big ( \sum_{j=1}^N {1 \over 1 - \cos \theta_j } \Big )^{h}
\Big \rangle_{(a ,b,\beta)}^{(\tilde{J})}.
\end{equation} 
The interest in this variant is that for $\beta = 2$ the PDF (\ref{Se3}) includes for particular $a,b$ the eigenvalue PDF
of the classical orthogonal and unitary symplectic matrix groups with Haar measure; see e.g.~\cite[Prop.~3.7.1]{Fo10}.
Thus then $\mathcal H_N^{(J)}(h) $ is proportional to the $2h$-th moment of the singular statistic $\sum_{j=1}^N 1/ (1 - \cos \theta_j)$; the
analysis of the moments for a related singular statistic of these classical groups has recently been the subject
of the study \cite{AS20}.

The characteristic function for the singular statistic relating to (\ref{Se2}) is
 \begin{equation}\label{Se5}
 \hat{P}_{N}^{(J)}(t) :=  \Big \langle \exp \Big ( i t \sum_{j=1}^N {1 \over u_j} \Big ) \Big \rangle_{({a},b,\beta)}^{(J)}.
\end{equation} 
To this author's knowledge this quantity, specialised to $\beta = 1$, first arose in a study of Davis \cite{Da68} into the distribution of
Hotelling's generalised $T_0^2$ statistic. Davis makes use of matrix differential equations to encode recursive
properties of the power series expansion of  $ \hat{P}_{N}^{(J)}(t)$ and the corresponding distribution function. At a practical level for applications in 
statistics, this leads to an efficient and accurate computational scheme to numerically evaluate the PDF of the distribution for fixed values of
the parameters, including $N$. However it does not give information on the large $N$ scaled limit.

In the study of (\ref{1.7c}), enabling progress in analysing the large $N$ scaled limit at the level of the underlying
distribution function, has been the analytic structure
of (\ref{2.30a}) in the case $s$ a positive integer. Consideration of the $N = 1$ case, when after a change of variables the problem
reduces to computing the  characteristic function of the univariate $F$-distribution \cite{Ph82},
shows that no such property holds
for (\ref{Se5}) (here the role of $s$ might be expected to be played by $a$), so computation of the explicit form of the limiting scaled distribution of  $\sum_{j=1}^N 1/u_j$ seems intractable.
On the other hand, some progress is possible at the level of moments, allowing (\ref{Se2}) to be computed for $h$ a non-negative
integer such that it is well defined.  An explicit formula for the mean is in fact
already known from previous literature. Thus we have \cite[Eqns.~(33), (34), (36) with $k=\lambda_1=1$]{FLD16}
(see also \cite{MRW15})
 \begin{equation}\label{FD}
  \Big \langle  \sum_{j=1}^N {1 \over u_j}  \Big \rangle_{({a},b,\beta)}^{(J)} = N {{a} + b + 1 + \beta (N - 1)/2 \over \hat{a}}.
\end{equation}  
This can be generalised, allowing for the computation of the scaled large $N$ limit.

  \begin{proposition}\label{P4.1}
 Let $p \in \mathbb Z^+$ with the further constraint $p < {a} +1$. Let
  $\ell(\kappa)$ denote the number of nonzero parts in $\kappa$. We have
    \begin{multline}\label{B1}
  \Big \langle  \Big ( \sum_{j=1}^N {1 \over u_j} \Big )^p \Big \rangle_{({a},b,\beta)}^{(J)} =
   \sum_{|\kappa| = p} 
   C_\kappa^{(2/\beta)}((1)^N) \\
    \times   \prod_{j=0}^{\ell(\kappa) - 1} { \Gamma( {a} + \beta j / 2 + 1 - \kappa_{j}) \over \Gamma( {a} + \beta j / 2 + 1)} 
    {     \Gamma( {a} + \beta ( N +j - 1)/ 2 + 2)  \over  \Gamma( {a} + \beta(N+ j-1) / 2 + 2 - \kappa_{j}) }.
    \end{multline}
    Consequently, in the notation of (\ref{2.6+}) with $\alpha = 2/\beta$,
    \begin{equation}\label{B2}
   \lim_{N \to \infty} {1 \over N^{2p} }  \Big \langle  \Big ( \sum_{j=1}^N {1 \over u_j} \Big )^p \Big \rangle_{({a},b,\beta)}^{(J)} =
      \sum_{|\kappa| = p}  |\kappa|!   {1 \over d_\kappa' h_\kappa}
        \prod_{j=0}^{\ell(\kappa) - 1} { \Gamma( {a} + \beta j / 2 + 1 - \kappa_{j}) \over \Gamma( {a} + \beta j / 2 + 1)}. 
   \end{equation}       
   \end{proposition}
  
  \begin{proof}
  Let $p$ be a non-negative integer. According to (\ref{J3Z}) we have
   \begin{equation}\label{C21}
   \Big \langle  \Big ( \sum_{j=1}^N {1 \over u_j} \Big )^p \Big \rangle_{({a},b,\beta)}^{(J)} =
   \sum_{|\kappa| = p} 
    \Big \langle   C_\kappa^{(2/\beta)}\Big ( {1 \over u_1}, \dots, {1 \over u_N} \Big )    \Big  \rangle_{({a},b,\beta)}^{(J)}.
    \end{equation}
 We know \cite[Exercises 12.1 q.1 \& q.2]{Fo10} that for $q \in \mathbb Z^+$, $q \ge \kappa_1$
   \begin{equation}\label{C22}   
   C_\kappa^{(2/\beta)}\Big ( {1 \over u_1}, \dots, {1 \over u_N} \Big )   =
   C_{-\kappa^R}^{(2/\beta)}(u_1,\dots,u_N) = \prod_{l=1}^N u_l^{-q} C_{(q)^N - \kappa^R}(u_1,\dots,u_N).
  \end{equation}
  In (\ref{C22}), $- \kappa^R$ denotes the partition $\kappa$, including the zero parts, with order reversed and each
  part multiplied by $-1$. So, for example, if $N=4$ and $\kappa = (2,1,0,0)$ then $-\kappa^R = (0,0,-1,-2)$ and
  $(q)^N - \kappa^R = (q,q,q-1,q-2)$.
  
  According to (\ref{15.pse1})
  \begin{multline}\label{C23-}
  \Big \langle   \prod_{l=1}^N u_l^{-q} C_{(q)^N - \kappa^R}(u_1,\dots,u_N)   \Big \rangle_{({a},b,\beta)}^{(J)} = {S_N({a} - q, b, \beta) \over
 S_N({a} , b, \beta) } \\
 \times   C_\kappa^{(2/\beta)}((1)^N)  {[ {a} - q + \beta (N - 1)/2  + 1]_{(q)^N - \kappa^R}^{(2/\beta)} \over
[ {a} - q +b +  \beta (N - 1)  + 2]_{(q)^N - \kappa^R}^{(2/\beta)}  }.
\end{multline}
Here use has been made of the fact that $ C_{-\kappa^R}^{(2/\beta)}((1)^N)  =  C_\kappa^{(2/\beta)}((1)^N)$ as follows from
(\ref{C22}).
From the definition (\ref{C0}) we can check that
  \begin{equation}\label{C23}   
  \prod_{j=0}^{N-1} { \Gamma ( {a} - q + 1 + j \beta/ 2) \over    \Gamma ( {a}  + 1 + j \beta/ 2)  } \,
  [ {a} - q +b +  \beta (N - 1)  + 2]_{(q)^N - \kappa^R}^{(2/\beta)} =
  \prod_{j=0}^{\ell(\kappa) - 1} { \Gamma( {a} + \beta j / 2 + 1 - \kappa_{N-j}^R) \over \Gamma( {a} + \beta j / 2 + 1)}.
  \end{equation}
The relevance of the first ratio of gamma functions comes from making use of
(\ref{3.2g}) in relation to the ratio of Selberg integrals in (\ref{C23-}).
  Using this as written, and also with ${a} \mapsto {a} + \beta (N - 1)/2 + 1$, and noting too that $\kappa_{N-j}^R = \kappa_j$ we see that (\ref{C23}) reduces to
  (\ref{B1}).
  
  Using the standard asymptotic formula for the ratio of gamma functions, and making use too of (\ref{2.6+}) and
 (\ref{2.6a+}), we see that  (\ref{B2}) is a direct consequence of (\ref{B1}).

 \end{proof}     
 
 \begin{remark}
 1.~The averages $  \Big \langle  \Big ( \sum_{j=1}^N  u_j \Big )^p \Big \rangle_{({a},b,\beta)}^{(J)}$, which in words are the
 moments of the trace in the Jacobi $\beta$-ensemble, have been studied in the recent work \cite{FK20}, and the analogue
 of (\ref{B1}) has been given. In addition, it has been shown that these moments satisfy a matrix recurrence with respect to $p$, where the matrices
 involved are of size $(N + 1) \times (N + 1)$. For the averages  (\ref{B1}), the results in \cite{Da68} imply an analogous matrix recurrence, at least
 for the case $\beta = 1$. \\
 2~Denote the LHS of (\ref{B2})  by $\mu_{\infty}^{(J)}(p)$, and set $\alpha = 2/ \beta$. Evaluation of the RHS shows
  \begin{equation}\label{C23+} 
  \mu_{\infty}^{(J)}(p) \Big |_{p=1} = {1 \over \alpha {a}},  \quad    \mu_{\infty}^{(J)}(p) \Big |_{p=2} =
 {1 \over \alpha^2 ( \alpha + 1) } {1 \over \hat{a} ({a} - 1)} +
 {1 \over \alpha ( \alpha + 1)} {1 \over \hat{a} ( {a} + 1/\alpha) }. 
 \end{equation} 
 In the case $\beta = 2$ (and thus $\alpha = 1$), these formulas are consistent with known results \cite[Eq.~(3.2)]{DFX19}.
 It is furthermore the case that for $\beta = 2$ the exponential generating function for the
 sequence of cumulants  $\{   \mu_{\infty}^{(J)}(p)  \}_{p=1}^\infty$ corresponding to the moments
  satisfies a nonlinear differential equation \cite{XDZ14,DFX19}. \\
 3.~Let $ \langle  \cdot  \rangle_{({a},\beta)}^{(L)}$ denote an average with respect to the Laguerre $\beta$-ensemble specified by
 the PDF proportional to
  \begin{equation}\label{L1k}
 \prod_{j=1}^N x_j^a  e^{- x_j}  |\Delta_N(\mathbf{x})|^{\beta}, \qquad x_j \in \mathbb R^+ \: (j=1,\dots,N)
 \end{equation}
 and define
  \begin{equation}\label{4.8d}
 \mu_{\infty}^{(L)}(p) := \lim_{N \to \infty} {1 \over N^p}   \Big \langle  \Big ( \sum_{j=1}^N  u_j \Big )^p \Big \rangle_{({a},\beta)}^{(L)}.
 \end{equation}
 The method of proof of Proposition \ref{P4.1} can be adapted to show
 $$
  \mu_{\infty}^{(L)}(p) = \mu_{\infty}^{(J)}(p).
  $$
  This is to be expected, as the singular linear statistic $ \sum_{j=1}^N  1/u_j$ probes the neighbourhood of the origin,
  which in both ensembles is characterised by the same factor $\prod_{l=1}^N u_l^a |\Delta(\mathbf a)|^\beta$ in the corresponding
  PDFs. The scaling $1/N^{2p}$ in (\ref{B2}), in contrast to the scaling $1/N^p$ in (\ref{4.8d}), is due to the typical spacing between
  eigenvalues in the neighbourhood of the origin being of order $1/N^2$ ($1/N$) for the Jacobi (Laguerre) ensemble. 
\end{remark}
 
 The result (\ref{B2}) from Proposition \ref{P4.1} allows for the scaled limit of (\ref{1.7+}) with $h$ a non-negative integer
 to be read off.

 \begin{corollary}\label{C4}
 Define $\mathcal H^{(J)}(h)$ by (\ref{1.7+}). In the notation of Proposition \ref{P4.1}, and with $h$ a non-negative integer
 such that $h < {a} +1$, we have
 \begin{equation}
 \lim_{N \to \infty}  {1 \over N^{2h}} \mathcal H_N^{(J)}(h) =  \sum_{|\kappa| = h}  |\kappa|!   {1 \over d_\kappa' h_\kappa}
        \prod_{j=0}^{\ell(\kappa) - 1} { \Gamma( {a} + \beta j / 2 + 1 - \kappa_{j}) \over \Gamma( {a} + \beta j / 2 + 1)}. 
 \end{equation}
 \end{corollary}

 \section*{Acknowledgements}
This work was supported by the Australian Research Council (ARC) Grant DP210102887 and
the ARC Centre of Excellence for Mathematical and Statistical Frontiers. The development of this
work has greatly benefitted from the constructive suggestions, the generous sharing of knowledge,
of the referees.

\appendix
\section*{Appendix}\label{A1}
\renewcommand{\thesection}{A} 
\setcounter{equation}{0}
For $\beta = 2$, when the exponential generating function (\ref{1.4a}) relates to 
determinantal structures such as (\ref{1.6}) and (\ref{1.6a}), there are known characterisations
in terms of Painlev\'e transcendents; see e.g.~\cite[Ch.~8\&9]{Fo10} for an account of results of
this type in random matrix theory. Here we want to use this characterisation to exhibit the small
$t$ analytic structure for continuous $s$. Fundamental to our exact calculations for positive integer
$s$ has been that this structure is given by an exponential times a power series in $t$. For $\beta = 2$,
we will see that for continuous $s$, an in general infinite number of fractional powers are also present.

The $s \in \mathbb Z^+$ determinantal evaluation (\ref{1.6a}) of the exponential generating function
(\ref{1.4a}) was first identified in the work \cite{CRS06}. Soon after, in \cite{FW06a}, it was observed that
up to an exponential factor, this same expression appears as a $\tau$-function in the Hamiltonian theory
of Painlev\'e III$'$ known from \cite{FW01a}. Thus an alternative evaluation of (\ref{1.6a}) is, under the
assumption $s \in \mathbb Z_{\ge 0}$,
\begin{equation}\label{1.6b}
\lim_{N \to \infty} {1 \over \tilde{C}_N} \mathcal H_N^{(Cy)}(s;t/2N) = e^{t/2} \tilde{E}_s^{\rm hard}(t), \qquad
 \tilde{E}_s^{\rm hard}(t) := \exp \Big ( - \int_0^{4t} ( \sigma_{\rm III'}(x) + s^2  ) \, {dx \over x} \Big ),
 \end{equation}
 where $ \sigma_{\rm III'}(x) =  \sigma_{\rm III'}(x;s) $ satisfies the particular $\sigma$-Painlev\'e III$'$ equation
 \begin{equation}\label{A.2a}
 (x \sigma_{\rm III'}'')^2 +  \sigma_{\rm III'}' (4   \sigma_{\rm III'} - 1) (   \sigma_{\rm III'} - x   \sigma_{\rm III'}') - {s^2 \over 16} = 0,
  \end{equation}
  subject to the small $x$ boundary condition
   \begin{equation}\label{A.2b}
   \sigma_{\rm III'}(x) = - s^2 + {x \over 8} + {\rm O}(x^2), \qquad s \in \mathbb Z_{\ge 0}.
  \end{equation}    
  
  In  \cite{FW06a} it was shown that $ \tilde{E}_s^{\rm hard}(t)$ also results as a scaling limit of the Laguerre unitary
  ensemble ($\beta = 2$ case of (\ref{L1k}))
    \begin{equation}\label{A.3}
   \tilde{E}_{N,s}(t)  \propto   \Big \langle   \prod_{l=1}^N (1 - \chi_{(0,t)}^{(l)}) ( x_l - t)^s  \Big \rangle_{(s,\beta)}^{(L)} \Big |_{\beta = 2},
   \end{equation}  
   where $ \chi_{(0,t)}^{(l)} = 1$ for $ x_l  \in (0,t)$, $ \chi_{(0,t)}^{(l)} = 0$ otherwise, and the proportionality is such that
   both sides equal unity when $t=0$. Thus
   \begin{equation}\label{A.4}  
     \tilde{E}_s^{\rm hard}(t)  = \lim_{N \to \infty} \tilde{E}_{N,s}(t/N). 
  \end{equation}       
    From the work of Winn \cite{Wi12} we know that (\ref{A.3}) is a transformed version of the exponential generating function
 (\ref{1.4a}) valid for continuous $s$. For continuous $s$, the quantity   $ \tilde{E}_{N,s}(t)$ was identified
 as a $\tau$-function in the Hamiltonian theory of Painlev\'e V in \cite{FW01a}, and its scaled limit was
 shown formally to give rise to the same $\sigma$-Painlev\'e III$'$ equation as for the $s \in \mathbb Z_{\ge 0}$ case,
 (\ref{A.2a}), and to the same functional form (\ref{1.6b}). Very recently a rigorous demonstration of the validity of
 (\ref{1.6b}) for continuous $s$ has been given in \cite{ABGS20a}.
 
 To make use of the knowledge that (\ref{1.6b}) remains valid for $s$ continuous, for our aim of determining the analytic form
 $ \tilde{E}_s^{\rm hard}(t)$,
 requires that the boundary condition
 (\ref{A.2b}) of the $\sigma$-Painlev\'e III$'$ equation (\ref{A.2b}) be appropriately extended.
 This problem was solved in \cite{FW06a}, by the strategy of 
 starting with (\ref{A.3}) rewritten as a certain Hankel determinant with elements given in
 terms of the hypergeometric ${}_1 F_1$ function, then expanding the latter for small $t$, which allows the small
 $t$ form of $ \tilde{E}_{N,s}(t)$ to be systematically determined. The small $t$ form of $  \tilde{E}_s^{\rm hard}(t) $
 then follows by taking the limit term-by-term. With the latter expressed in terms of $ \sigma_{\rm III'}(x)$ according to
 (\ref{1.6b}) 
this procedure shows  \cite[Eq.~(4.2) with $s \mapsto x$, $\mu = a = s$, $\xi = 1$]{FW06a}
   \begin{equation}\label{A.5}  
 \sigma_{\rm III'}(x) =  \Big ( - s^2 + {x \over 8} + {\rm O}(x^2) \Big ) + c_s  x^{2s+1} ( 1 + {\rm O}(x) ) + {\rm O}( x^{2 (2s+1)}),
  \end{equation} 
  where
    \begin{equation}\label{A.5a}    
    c_s := \Big ( \lim_{\epsilon \to 0} {\sin \pi s \over \sin (2 \pi s + \epsilon) } \Big ) {(\Gamma(s+1))^2 \over \Gamma(2(s+1)) \Gamma^2(2s+1)}.
  \end{equation}     
  The limit in (\ref{A.5a}) diverges for $2s$ half an odd positive integer. In these cases $\log x$ terms  appear in the functional form
  (\ref{A.5}).
  
  Substituting (\ref{A.5}) in the $\sigma$-Painlev\'e III$'$ equation (\ref{A.2b}) reveals that about the origin the
 transcendent $ \sigma_{\rm III'}(x)$, for  $2s$ not half an odd positive integer,
 exhibits the Puiseux-type series solution
   \begin{equation}\label{A.6} 
  \sigma_{\rm III'}(x) = \sum_{j=0}^\infty x^{j (2 s + 1)} g_j(x),
    \end{equation}  
    where each $  g_j(x)$ is a power series in $x$. Substituting in  (\ref{1.6b}) gives that this same analytic structure persists
    for $ \tilde{E}_s^{\rm hard}(t)$. Moreover, in the case that $s$ is a non-negative integer, we know from (\ref{A.2b}) that only
    the term $j=0$ in the sum contributes, showing how this case is special.


\providecommand{\bysame}{\leavevmode\hbox to3em{\hrulefill}\thinspace}
\providecommand{\MR}{\relax\ifhmode\unskip\space\fi MR }
\providecommand{\MRhref}[2]{%
  \href{http://www.ams.org/mathscinet-getitem?mr=#1}{#2}
}
\providecommand{\href}[2]{#2}

\end{document}